\documentclass
[
    a4paper,
    DIV=11,
    abstract=true,
]
{scrartcl}

\usepackage
{
    graphicx,
    amssymb,
    amsmath,
    amsthm,
    dsfont, 
    xcolor,
    mathtools,
    authblk,
    enumitem,
    nicefrac,
    array,
}

\usepackage[notrig]{physics}

\usepackage[utf8]{inputenc}

\usepackage[pdffitwindow=false,
            plainpages=false,
            pdfpagelabels=true,
            pdfpagemode=UseOutlines,
            pdfpagelayout=SinglePage,
            bookmarks=false,
            colorlinks=true,
            hyperfootnotes=false,
            linkcolor=blue,
            urlcolor=blue!30!black,
            citecolor=green!50!black]{hyperref}

\usepackage[bf,normal]{caption}

\graphicspath{{figs/}}

\newcommand{\subfiguretitle}[1]{{\scriptsize{#1}} \\}
\newcommand{\R}{\mathbb{R}}                                     
\newcommand{\pd}[2]{\frac{\partial#1}{\partial#2}}              
\newcommand{\ts}{\hspace*{0.1em}}                               
\providecommand{\abs}[1]{\left\lvert #1 \right\rvert}           
\providecommand{\norm}[1]{\left\lVert #1 \right\rVert}          

\newcolumntype{C}[1]{>{\centering\let\newline\\\arraybackslash\hspace{0pt}}m{#1}}

\newcommand\xqed[1]{\leavevmode\unskip\penalty9999 \hbox{}\nobreak\hfill \quad\hbox{#1}}
\newcommand{\exampleSymbol}{\xqed{$\triangle$}}

\newtheorem{theorem}{Theorem}[section]
\newtheorem{corollary}[theorem]{Corollary}
\newtheorem{lemma}[theorem]{Lemma}

\theoremstyle{definition}
\newtheorem{example}[theorem]{Example}
\newtheorem{remark}[theorem]{Remark}

\makeatletter
\renewcommand*\env@matrix[1][*\c@MaxMatrixCols c]{%
  \hskip -\arraycolsep
  \let\@ifnextchar\new@ifnextchar
  \array{#1}}
\makeatother

\mathtoolsset{centercolon} 

\allowdisplaybreaks

\DeclareCaptionLabelFormat{period}{#1~#2}
\makeatletter
\def\blfootnote{\gdef\@thefnmark{}\@footnotetext}
\makeatother

\begin{document}

\title{Koopman analysis of quantum systems}
\author[1]{Stefan Klus\footnote{These authors contributed equally to this work}}
\author[2,3]{Feliks Nüske$^\ast$}
\author[4]{Sebastian Peitz}
\affil[1]{Department of Mathematics, University of Surrey, UK}
\affil[2]{Department of Mathematics, Paderborn University, Germany}
\affil[3]{Max Planck Institute for Dynamics of Complex Technical Systems, Magdeburg, Germany}
\affil[4]{Department of Computer Science, Paderborn University, Germany}

\date{}

\maketitle

\begin{abstract}
Koopman operator theory has been successfully applied to problems from various research areas such as fluid dynamics, molecular dynamics, climate science, engineering, and biology. Applications include detecting metastable or coherent sets, coarse-graining, system identification, and control. There is an intricate connection between dynamical systems driven by stochastic differential equations and quantum mechanics. In this paper, we compare the ground-state transformation and Nelson's stochastic mechanics and demonstrate how data-driven methods developed for the approximation of the Koopman operator can be used to analyze quantum physics problems. Moreover, we exploit the relationship between Schrödinger operators and stochastic control problems to show that modern data-driven methods for stochastic control can be used to solve the stationary or imaginary-time Schrödinger equation. Our findings open up a new avenue towards solving Schrödinger's equation using recently developed tools from data science.
\end{abstract}

\section{Introduction}

Relationships between the Schrödinger equation and the Fokker--Planck equation have been explored since the early days of quantum mechanics. Schrödinger \cite{Schroedinger31} already wrote:
{ \setlength{\leftmargini}{2ex}
\begin{quote}
\textit{Eine gewisse Verwandtschaft der wellenmechanischen Grundgleichung und der Fokkerschen Gleichung, sowie der an beide anknüpfenden statistischen Begriffsbildungen hat sich wohl jedem aufgedrängt, der mit den beiden Ideenkreisen hinlänglich vertraut ist.}\footnote{A certain relationship between the basic wave mechanical equation and the Fokker equation, as well as the related statistical concepts, has probably struck anyone who is sufficiently familiar with both ideas.}
\end{quote}}
He then continues to point out two major differences, namely that (1) in classical systems the probability density $ \rho $ itself obeys a linear differential equation, whereas in quantum mechanics the wave function $ \psi $ (not the probability density $ \rho $) satisfies a related linear differential equation, from which we can obtain the probabilities by computing $ \rho = \abs{\psi}^2 $, and that (2) the imaginary unit~$ i $ changes the nature of the differential equation and makes it reversible, while in the classical case the dynamics are irreversible.

There are many different formulations and interpretations of quantum mechanics, a concise overview is given in \cite{Styer02}. The pilot-wave description by de Broglie--Bohm and Nelson's stochastic mechanics are often subsumed under \emph{hidden variable theories}, where the latter is less well-known and often presented as a stochastic variant of the former \cite{Bacciagaluppi05}. The main conceptual difference, however, is that in the original de Broglie--Bohm theory the dynamics depend on the phase of the wave function obeying the Schrödinger equation, whereas Nelson aimed at deriving the wave function and the Schrödinger equation, assuming only that particle trajectories can be described by diffusion processes in configuration space \cite{Bacciagaluppi99, Bacciagaluppi05}. While conventional quantum mechanics and stochastic mechanics make the same predictions, the wave function $ \psi $ plays no fundamental role in the context of Nelson's formulation and can be viewed as a convenient computational device \cite{Carlen06}. Nevertheless, it is also possible to obtain the associated stochastic process given a wave function $ \psi $ that solves the Schrödinger equation. We will use this viewpoint as a convenient tool to derive stochastic dynamics. A comprehensive review and comparison of hidden variable theories and their interpretations can be found in \cite{Gouesbet14}.

We will not discuss the physical interpretation of hidden variables, our goal instead is to establish and exploit mathematical connections between stochastic processes---in particular classical drift-diffusion processes---and quantum mechanics. In recent years, a wealth of numerical methods has been developed in order to analyze stochastic dynamics based on trajectory data, see \cite{Prinz2011, Noe2013, WKR15, KKS16, KNKWKSN18, KNPNCS20}. Moreover, we showed in \cite{KNH20} that data-driven methods for the analysis of classical dynamical systems can be applied either directly to quantum systems or to their stochastic counterparts. To this end, we used a well-known transformation of the Schrödinger equation into a Kolmogorov backward equation, which is governed by the generator of an associated stochastic differential equation. This transformation is based on the assumption that a strictly positive ground state exists. Equivalent transformations exist also for the Fokker--Planck equation \cite{Risken89, Pav14}. These methods were then used to compute eigenfunctions of the time-independent Schrödinger equation.

In this paper, we will derive the transformation of the Schrödinger equation to the Kolmogorov equation in a more general context and show how the transformation leads to numerical methods to solve the Schrödinger equation based on data. We will present and illustrate various algorithms in detail, requiring different amounts of prior information. The main contributions of this work are:
\begin{itemize}[wide, itemindent=\parindent, itemsep=0ex, topsep=0.5ex]
\item We revisit the connections between Koopman operator theory and quantum mechanics using the transformation of the Schrödinger equation into a Kolmogorov backward equation. We compare the transformation to Nelson's stochastic formulation and explain how this transformation can be used to obtain multiple new solutions of Schrödinger's equation based on trajectory data, provided a single solution is known.
\item We recapitulate the relationships between solutions of the imaginary-time Schrödinger equation and a stochastic optimal control problem based on the Feynman--Kac formula and Hamilton--Jacobi--Bellman equation. We then use recent results for data-driven control of control-affine systems to show that the Schrödinger equation can be solved by means of an optimal control problem constrained by an ordinary differential equation. This formulation does not require any prior knowledge of a solution.

\item We show how to apply recently developed data-driven techniques for the approximation of the Koopman operator to quantum systems. In the quantum context, this allows us to compute ground and excited states as well as general solutions of the time-dependent Schrödinger equation.
\end{itemize}

In Section~\ref{sec:Koopman and Schroedinger}, we will introduce the Koopman operator and its generator as well as the Schrödinger operator. In Section~\ref{sec:Stochastic descriptions}, hidden variable interpretations of quantum mechanics, which we will use for the numerical analysis of quantum systems, will be outlined. Furthermore, we will explore connections between the ground-state transformation and Nelson's stochastic mechanics and derive stochastic descriptions of well-known quantum systems. An optimal control formulation for the Schrödinger equation will be derived in Section~\ref{sec:Optimal control formulation}. Section~\ref{sec:Data-driven methods for quantum systems} shows how Koopman operator theory can be applied to quantum physics problems and Section~\ref{sec:numerics_control} how the control formulation can be used to solve the imaginary-time Schrödinger equation. Open questions and future work will be discussed in Section~\ref{sec:Conclusion}.

\section{Koopman operator theory and quantum mechanics}
\label{sec:Koopman and Schroedinger}

We start by briefly introducing the stochastic Koopman operator and the Schrödinger operator. The goal is to apply Koopman operator theory and related numerical methods to quantum systems.

\subsection{The Koopman operator}
\label{subsec:koopman_operator}
Consider a dynamical system defined on a state space $ \mathbb{X} \subset \mathbb{R}^d $, given by a stochastic differential equation
\begin{equation} \label{eq:SDE}
    \mathrm{d}X_t = b(X_t, t) \ts \mathrm{d}t + \sigma(X_t, t) \ts \mathrm{d}B_t.
\end{equation} 
Here, $ b $ is called the \emph{drift term}, $ \sigma $ the \emph{diffusion term}, and $ B_t $ is a standard Wiener process. We also introduce the notation $a(x, t) := \sigma(x, t)\sigma(x, t)^\top$ for the covariance matrix of the diffusion. The \emph{Koopman operator} \cite{Ko31, LaMa94, BMM12, KKS16} is then defined, for any time lag $ \Delta_t \geq 0 $ and suitable functions $ f $, by the conditional expectation:
\begin{equation} \label{eq:Koopman operator}
    (\mathcal{K}^{t, \Delta_t} f)(x) = \mathbb{E}^x[f(X_{t + \Delta_t})] = \mathbb{E}[f(X_{t + \Delta_t})\, \vert X_t = x \,].
\end{equation}
The Koopman operator for ordinary differential equations, i.e., $ \sigma \equiv 0 $, can be regarded as a special case where the computation of the expectation value can be omitted. The family of operators $ \mathcal{K}^{t, \Delta_t} $ satisfy the important semigroup property
\begin{equation*}
    \mathcal{K}^{t, \Delta_{t_1} + \Delta_{t_2}} = \mathcal{K}^{t+\Delta_{t_1}, \Delta_{t_2}}\mathcal{K}^{t, \Delta_{t_1}}.
\end{equation*}
This motivates the definition of the associated \emph{generator} of the Koopman operator \eqref{eq:Koopman operator} as the time derivative $ \mathcal{L}_t f = \lim\limits_{\Delta_t \rightarrow 0} \frac{1}{\Delta_t} \left(\mathcal{K}^{t, \Delta_t} f - f \right) $. The result is the following second order differential operator, also known as Kolmogorov backward operator:
\begin{equation} \label{eq:Koopman generator}
\begin{split}
    \mathcal{L}_t f(x) &= \sum_{i=1}^d b_i(x, t) \ts \pd{f(x)}{x_i} + \frac{1}{2} \sum_{i=1}^d \sum_{j=1}^d a_{ij}(x, t) \ts \pd{^2 f(x)}{x_i \ts \partial x_j} \\ &= b(x, t) \cdot \nabla f(x) + \frac{1}{2} a(x, t) : \nabla^2 f(x).
\end{split}
\end{equation}
The adjoint of the Koopman operator is called \emph{Perron--Frobenius operator}, its generator is the \emph{Fokker--Planck operator}
\begin{equation*}
    \mathcal{L}^*_t p(x) = -\sum_{i=1}^d \pd{(b_i(x, t) p(x))}{x_i} + \frac{1}{2} \sum_{i=1}^d \sum_{j=1}^d \pd{^2  (a_{ij}(x, t) p(x))}{x_i \ts \partial x_j}.
\end{equation*}
The associated Fokker--Planck equation describes the time evolution of the probability distribution of the process $X_t$. If the drift term $ b $ is defined by the gradient of a scalar potential $ V $, the diffusion $ \sigma $ is isotropic, and both are time-independent, we obtain the simplified stochastic differential equation
\begin{equation} \label{eq:drift-diffusion process}
    \mathrm{d}X_t = -\nabla V(X_t) \ts \mathrm{d}t + \sqrt{2 \beta^{-1}} \ts \mathrm{d}B_t,
\end{equation}
which is frequently used in molecular dynamics and, as we will see below, can also be used to model quantum systems. The parameter $ \beta $ is called \emph{inverse temperature}. The smaller $ \beta $, the larger the noise. For systems of this form, the generator is given by
\begin{equation*} 
    \mathcal{L} f = -\nabla V \cdot \nabla f + \beta^{-1} \Delta f.
\end{equation*}

Koopman operators provide a lifting of the nonlinear dynamical system~\eqref{eq:SDE} into the infinite-dimensional space of observable functions $ f $, where the dynamics are driven by the linear operator $ \mathcal{L}_t $. In practice, $ \mathcal{K}^{t, \Delta_t} $ must be approximated on a finite-dimensional subspace: For linearly independent functions $ \Phi = [\phi_1, \dots, \phi_n]^\top $, the Galerkin projection of the Koopman operator is given by the matrix:
\begin{align*}
    K^{t, \Delta_t} &= \left(C^t\right)^{-1}\ts A^t, & C^t_{ij} &= \mathbb{E}\left[\phi_i(X_t) \ts \phi_j(X_t)\right], & A^t_{ij} &= \mathbb{E}\left[\phi_i(X_t) \ts \phi_j(X_{t+\Delta_t})\right].
\end{align*}
Computing the matrix approximation $ K^{t, \Delta_t} $ thus only requires instantaneous and one-step correlation functions, which can be efficiently approximated using simulation data~\cite{NKPMN14, WKR15, KKS16, KNKWKSN18}. By definition, the matrix approximation $K^{t, \Delta_t}$ can be used to approximately predict the expectation at time $t+\Delta_t$ for observables contained in the linear span of the functions $ \phi_1, \dots, \phi_n $.

\subsection{The Schrödinger equation}

We will use atomic units throughout the paper for the sake of simplicity. Given the Hamiltonian $ \mathcal{H} = - \frac{1}{2} \Delta + W $, where $ W $ is the potential energy,\!\footnote{In order to avoid confusion, molecular dynamics potentials are denoted by $ V $ and quantum mechanics potentials by $ W $.} the \emph{time-dependent Schrödinger equation} is given by
\begin{equation} \label{eq:TDSE}
    i \pd{\psi}{t} = \mathcal{H} \psi.
\end{equation}
This partial differential equation describes the evolution of the wave function $ \psi $. The \emph{time-independent Schrödinger equation}, defined by
\begin{equation} \label{eq:TISE}
    \mathcal{H} \psi = E \psi,
\end{equation}
is an eigenvalue problem, where $ E $ is the associated energy. Given a solution $ \psi $ of~\eqref{eq:TISE}, the corresponding time-dependent solution of \eqref{eq:TDSE} is
\begin{equation} \label{eq:TDSE solution}
    \psi(x, t) = \psi(x) \ts e^{-i E t}.
\end{equation}
We refer the reader to \cite{Hall13} for a detailed introduction. By setting $ \tau = i \ts t $, we obtain the so-called \emph{imaginary-time Schrödinger equation}
\begin{equation} \label{eq:it TDSE}
    \pd{\psi}{\tau} = -\mathcal{H} \psi,
\end{equation}
which, for instance, plays an important role in quantum Monte Carlo methods. These different types of Schrödinger equations will be used below to derive stochastic models and to generate data.

\section{Stochastic descriptions of quantum systems}
\label{sec:Stochastic descriptions}

There are different ways to derive stochastic descriptions of quantum systems. In this section, we will introduce the ground-state transformation, outline Nelson's stochastic mechanics, and compare the resulting models.

\subsection{Ground-state transformation}
\label{subsec:ground_state_transform}

In \cite{KNH20}, we considered quantum systems with strictly positive ground states. The transformation to the Kolmogorov backward equation, however, also works for complex wave functions. Since our goal is to compare this transformation with Nelson's stochastic mechanics, we directly use his notation and write $ \psi_0 = e^{R + i \ts S} $.

\begin{lemma} \label{lem:Ground-state transformation}
Assume that $ \psi_0 $ is a stationary solution with $ \mathcal{H} \psi_0 = E_0 \ts \psi_0 $. Let $ \psi_0(x) \ne 0 $ for all $ x $, then $ \mathcal{H} (\psi_0 \ts f) = E (\psi_0 \ts f) $ if and only if $ -\mathcal{L} f = (E - E_0) f $, where
\begin{align*}
    \mathcal{L} f = \big(\nabla R + i \ts \nabla S \big) \cdot \nabla f + \frac{1}{2} \Delta f.
\end{align*}
\end{lemma}

\begin{proof}
First, we compute $ \Delta (\psi_0 \ts f) = \Delta \psi_0 \ts f + 2 \ts \nabla \psi_0 \cdot \nabla f + \psi_0 \ts \Delta f $. Then
\begin{align*}
    (\mathcal{H} - E) (\psi_0 \ts f)
        &=  -\frac{1}{2} \Delta (\psi_0 \ts f) + W(\psi_0 \ts f) - E(\psi_0 \ts f) \\
        &= -\frac{1}{2} (\Delta \psi_0 \ts f + 2 \ts \nabla \psi_0 \cdot \nabla f + \psi_0 \ts \Delta f) + W(\psi_0 \ts f) - E(\psi_0 \ts f) \\
        &= -\frac{1}{2} (2 \ts \nabla \psi_0 \cdot \nabla f + \psi_0 \ts \Delta f) - (E - E_0)(\psi_0 f).
\end{align*}
In the last step, we used $ (\mathcal{H} - E_0) \psi_0 = -\frac{1}{2} \Delta \psi_0 + W \psi_0 - E_0 \ts \psi_0 = 0 $. Dividing by $ \psi_0 $ and using $ \psi_0^{-1} \nabla \psi_0 = \nabla R + i \nabla S $, the equation is only zero if
\begin{equation*}
    -\mathcal{L} f = -\big(\nabla R + i \ts \nabla S\big) \cdot \nabla f - \frac{1}{2} \Delta f = (E - E_0) f. \qedhere
\end{equation*}
\end{proof}

Note that we in general obtain a complex-valued partial differential equation, but as a special case the ground-state transformation used, e.g., in \cite{Risken89, Pav14}.

\begin{corollary}
Assuming the ground state is strictly positive, i.e., $ S \equiv 0 $, this yields
\begin{equation*}
    \mathcal{L} f = \nabla R \cdot \nabla f + \frac{1}{2} \Delta f
\end{equation*}
and $ \mathcal{L} $ is the generator of a drift-diffusion process \eqref{eq:drift-diffusion process} with potential $ V(x) = -R(x) $ and diffusion constant $ \beta^{-1} = \frac{1}{2} $.
\end{corollary}

The corollary implies that we can compute excited states of the quantum system by computing eigenfunctions of the Koopman operator associated with the stochastic differential equation. This will be analyzed in more detail below. Furthermore, if we instead solve $ \mathcal{H} (\psi_0^{-1} \ts f) = E (\psi_0^{-1} \ts f) $, we obtain an operator
\begin{align*}
    \mathcal{L}^* f = -\big(\Delta R + i \ts \Delta S \big) f - \big(\nabla R + i \ts \nabla S \big) \cdot \nabla f + \frac{1}{2} \Delta f.
\end{align*}
If $ S \equiv 0 $, this is the standard Fokker--Planck equation for the drift-diffusion process derived above, whose invariant density is $ \rho_0 = e^{2 \ts R} $. For the complex-valued case, we obtain an eigenfunction of the form $ \rho_0 = e^{2 \ts (R + i \ts S)} $. Before moving on, we make the important observation that the ground-state transformation is not limited to the stationary case.

\begin{lemma} \label{lem:Transformation TDSE}
We can also apply the transformation used in Lemma~\ref{lem:Ground-state transformation} to the time-dependent Schrödinger equation:
\begin{enumerate}[wide, itemindent=\parindent, itemsep=0ex, topsep=0.5ex, label=\roman*)]
\item Assume that $ \psi_0 $ is a solution of the time-dependent Schrödinger equation \eqref{eq:TDSE}. Then it holds that $ i \pd{}{t} (\psi_0 \ts f) = \mathcal{H} (\psi_0 \ts f) $ if and only if $ i \pd{f}{t} = -\mathcal{L} f $.
\item For the imaginary-time Schrödinger equation \eqref{eq:it TDSE} it holds that $ \pd{}{\tau} (\psi_0 \ts f) = -\mathcal{H} (\psi_0 \ts f) $ if and only if $ \pd{f}{\tau} = \mathcal{L} f $.
\end{enumerate}
\end{lemma}

\begin{proof}
The only difference is that we now have to compute $ \pd{}{t}(\psi_0 \ts f) = \pd{\psi_0}{t} f + \psi_0 \ts \pd{f}{t} $. For the second part, we set $ \tau = i \ts t $.
\end{proof}

\begin{example} \label{ex:1D problems}
As an illustration, let us consider a few one-dimensional problems that can be solved analytically and derive stochastic models. An overview of the potentials $ W $ and the corresponding potentials $ V $ of the associated drift-diffusion processes is shown in Figure~\ref{fig:Comparison}.

\begin{figure}
    \definecolor{matlab1}{RGB}{0, 114, 189}
    \definecolor{matlab2}{RGB}{217, 83, 25}
    \definecolor{matlab3}{RGB}{237, 177, 32}
    \definecolor{matlab4}{RGB}{126, 47, 142}
    \newcommand{\cdash}[1]{\textcolor{#1}{\rule[0.5ex]{1em}{0.3ex}}}
    \centering
    \begin{minipage}{0.35\textwidth}
        \centering
        \subfiguretitle{(a)}
        \includegraphics[width=\textwidth]{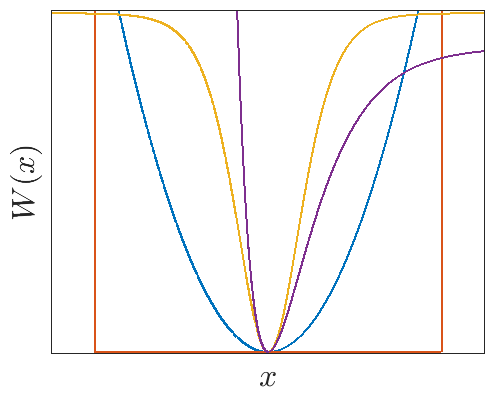}
    \end{minipage}
    \hspace*{3ex}
    \begin{minipage}{0.35\textwidth}
        \centering
        \subfiguretitle{(b)}
        \includegraphics[width=\textwidth]{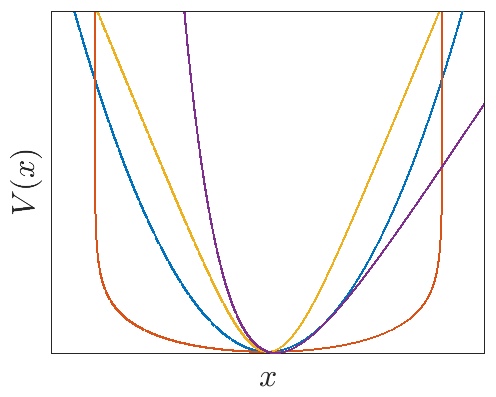}
    \end{minipage}
    \caption{(a) Potentials $ W $ of analytically solvable one-dimensional quantum systems, where \cdash{matlab1} represents the quantum harmonic oscillator, \cdash{matlab2} the particle in a box, \cdash{matlab3} the Pöschl--Teller potential, and \cdash{matlab4} the Morse potential. (b)~Potentials $ V $ of the corresponding drift-diffusion processes. Some functions were shifted for the sake of comparison. The Morse potential is the only non-symmetric system.}
    \label{fig:Comparison}
\end{figure}

\begin{enumerate}[wide, itemindent=\parindent, itemsep=0ex, topsep=0.5ex, label=\roman*)]
\item The correspondence between the \emph{quantum harmonic oscillator} and the Ornstein--Uhlen\-beck process is well-known, see, e.g., \cite{Risken89, Pav14}. The potential of the system is given by $ W(x) = \frac{1}{2} \ts \omega^2 x^2 $, with angular frequency $ \omega $. The eigenvalues of the Schrödinger operator are given by $ E_\ell = \omega \left(\ell + \frac{1}{2}\right) $ and the eigenfunctions by
\begin{equation*}
    \psi_\ell(x) = \tfrac{1}{\sqrt{2^\ell\,\ell!}} \left(\tfrac{\omega}{\pi}\right)^{\nicefrac{1}{4}} e^{-\frac{\omega \ts x^2}{2}} H_\ell \left(\sqrt{\omega} \ts x \right),
\end{equation*}
where $ H_\ell $ is the $\ell$th physicists' Hermite polynomial. The corresponding stochastic differential equation is
\begin{equation*}
    \mathrm{d}X_t = -\omega \ts X_t \ts \mathrm{d}t + \mathrm{d}B_t.
\end{equation*}
The eigenvalues and eigenfunctions of the associated Koopman generator are $ \lambda_\ell = -\omega \ts \ell $ and $ \varphi_\ell(x) = H_\ell\left(\sqrt{\omega} \ts x\right) $. Using Lemma~\ref{lem:Ground-state transformation}, we obtain $ -\lambda_\ell = E_\ell - E_0 $ so that $ E_\ell = \frac{1}{2} \ts \omega + \omega \ts \ell $. This illustrates how stochastic differential equations can be used to compute higher-energy states, see \cite{KNH20} for a detailed analysis.

\item For the \emph{particle in a box}, the potential is defined by
\begin{equation*}
    W(x) =
    \begin{cases}
        0, & 0 \le x \le L, \\
        \infty, & \text{otherwise}.
    \end{cases}
\end{equation*}
We obtain the eigenvalues $ E_\ell = \frac{\pi^2 (\ell+1)^2}{2 L^2} $ and the eigenfunctions $ \psi_\ell = \sqrt{\frac{2}{L}} \sin\left( \frac{\pi (\ell + 1)}{L} \ts x \right) $, for $ \ell = 0, 1, 2, \dots $, see \cite{Risken89, Hall13}. Hence, $ R(x) = \log \left( \sin \left( \frac{\pi}{L} x \right) \right) $ for $ x \in (0, L) $ so that
\begin{equation*}
    \mathrm{d}X_t = \tfrac{\pi}{L} \cot \left( \tfrac{\pi}{L} x \right) \ts \mathrm{d}t + \mathrm{d}B_t.
\end{equation*}
Note that $ R(x) \to -\infty $ for $ x \to 0^+ $ and $ x \to L^- $.

\item The \emph{Pöschl--Teller potential} is defined by
\begin{equation*}
    W_s(x) = -\frac{s(s+1)}{2} \sech^2(x),
\end{equation*}
where $ s \in \mathbb{N} $ is a parameter determining its depth \cite{BH18:PoeschlTeller}. The ground state is $ \psi_0(x) = \sech^s(x) $ with $ E_0 = -\frac{s^2}{2} $. Transforming this system, we obtain $ R(x) = s \ts \log(\sech(x)) $ and the drift-diffusion process
\begin{equation*}
    \mathrm{d}X_t = -s \ts \tanh(X_t) \ts \mathrm{d}t + \mathrm{d}B_t.
\end{equation*}
Numerical results will be presented in Section~\ref{sec:Data-driven methods for quantum systems}.
\item Also for the \emph{Morse potential}
the transformations can be carried out analytically. The resulting potential $ V $ for the drift-diffusion process is shown in Figure~\ref{fig:Comparison}. \exampleSymbol
\end{enumerate}

\end{example}

The fact that every analytically solvable Schrödinger equation can also serve as an example of a solvable Fokker--Planck equation was already utilized in \cite{Risken89}. Conversely, given a drift-diffusion process with potential $ V $, inverse temperature $ \beta^{-1} = \frac{1}{2} $, and unique invariant density $ \rho_0 = e^{-\beta V} $, then
\begin{equation*}
    \rho_0^{\nicefrac{1}{2}} \mathcal{L}\big(\rho_0^{-\nicefrac{1}{2}} \psi\big)
        = -\left[\left(\frac{\beta}{4} \abs{\nabla V}^2 - \frac{1}{2} \Delta V\right) \psi - \beta^{-1} \Delta \psi\right]
        = -\mathcal{H} \psi,
\end{equation*}
is a Schrödinger operator with potential
\begin{equation*}
    W = \frac{\beta}{4} \abs{\nabla V}^2 - \frac{1}{2} \Delta V.
\end{equation*}
Note that the potential $ W $ is defined in such a way that the ground state energy is zero.

\begin{remark} \label{rem:Invariant potentials}
An interesting question then is which potentials are invariant under this transformation. Using the ansatz $ V(x) = \frac{1}{2} x^\top A x + b^\top x + c $, we obtain, e.g., solutions of the form $ A = 2 \ts \beta^{-1} I $,  $ b $ can be arbitrary, and $ c = \frac{\beta}{4} b^\top b - \frac{1}{2} \trace(A) = \frac{\beta}{4} b^\top b - \beta^{-1} d $, where $ d $ is the dimension of the problem. This corresponds to $ d $ uncoupled quantum harmonic oscillators (or Ornstein--Uhlenbeck processes) whose potentials might be shifted along the $ x $ and $ y $ axes.
\end{remark}

\subsection{Nelson's stochastic mechanics}
\label{subsec:nelson_mechanics}

Nelson's stochastic mechanics~\cite{Nelson66} is one of the most general stochastic formulations of quantum mechanics. The idea is to determine a real-valued stochastic differential equation such that the distribution of $ X_t $ equals the quantum probability distribution $ \rho(x, t) = \abs{\psi(x, t)}^2 $. While Nelson originally aimed at deriving quantum mechanics assuming only that particles follow diffusion processes in configuration space \cite{Bacciagaluppi99, Bacciagaluppi05}, the stochastic dynamics can also be obtained from a complex-valued solution of the time-dependent Schrödinger equation~\eqref{eq:TDSE}, which Nelson described as ``quantum mechanics made difficult''~\cite{Nelson85}. In this section, we will assume that a reference solution $ \psi $ is known. Given a solution $ \psi = e^{R + i \ts S} $, we define
\begin{equation*}
    u = \nabla R
    \quad \text{and} \quad
    v = \nabla S,
\end{equation*}
where $ u $ is called \emph{osmotic velocity} and $ v $ \emph{current velocity}. According to the theory of Brownian motion, $ u $ is the velocity acquired by a Brownian particle that is in equilibrium with respect to an external force, balancing the osmotic force \cite{Nelson66}. As outlined in Appendix~\ref{app:Continuity equation}, the quantum probability $ \rho $ satisfies the continuity equation
\begin{equation} \label{eq:continuity}
    \pd{\rho}{t} = - \nabla \cdot j,
\end{equation}
with the probability current $ j = v \ts \rho $. This justifies the names osmotic and current velocity for $ u $ and $ v $, respectively. Based on these velocity definitions, Nelson constructs a stochastic differential equation as follows: The drift term, which is called \emph{mean forward velocity} in this context, is given by
\begin{equation*}
    b = u + v
\end{equation*}
and the diffusion term is defined by $ \sigma = 1 $, which implies that $ \beta^{-1} = \frac{1}{2} $. This system can be viewed as a drift-diffusion process with corresponding potential $ V = -R - S $.

\begin{remark} \label{SDE equivalence}
A few remarks are in order:
\begin{enumerate}[wide, itemindent=\parindent, itemsep=0ex, topsep=0.5ex, label=\roman*)]
\item Several equivalent expressions for $ u $ and $ v $ are used in the literature. We could also define
\begin{equation*}
    u = \Re\left(\frac{\nabla \psi}{\psi}\right) = \Re(\nabla \log \psi)
    \quad \text{and} \quad
    v = \Im\left(\frac{\nabla \psi}{\psi}\right) = \Im(\nabla \log \psi).
\end{equation*}
Bacciagaluppi \cite{Bacciagaluppi05}, on the other hand, writes $ \psi = \widetilde{R} \ts e^{i \ts S} $ so that $ u = \frac{1}{2} \frac{\nabla \widetilde{R}^2}{\widetilde{R}^2} = \frac{\nabla \widetilde{R}}{\widetilde{R}} = \nabla \log \widetilde{R} $.
\item While in Nelson's original derivation $ \beta^{-1} = \frac{1}{2} $, it is possible to construct alternative approaches with arbitrary positive diffusion constants as discussed in \cite{Bacciagaluppi05}, see also Appendix~\ref{app:Continuity equation}.
\end{enumerate}
\end{remark}

\subsection{Comparison}

The ground-state transformation (Section~\ref{subsec:ground_state_transform}) and Nelson's mechanics (Section~\ref{subsec:nelson_mechanics}) both provide stochastic differential equations based on a given wave function $\psi$. However, they differ with regard to the specific settings where they can be meaningfully applied and with regard to the kind of additional information that can be gained from analyzing these stochastic systems. We summarize the differences as follows:

\begin{enumerate}[wide, itemindent=\parindent, itemsep=0ex, topsep=0.5ex, label=\roman*)]
\item The ground-state transformation can be applied to a given positive solution of the stationary Schrödinger equation~\eqref{eq:TISE} or of the imaginary-time equation~\eqref{eq:it TDSE}. In either case, additional solutions of the same Schrödinger equation can be determined by solving the corresponding Kolmogorov equation, which can be achieved using Koopman operator based methods. For the time-dependent Schrödinger equation in real time~\eqref{eq:TDSE}, the stochastic differential equation becomes complex-valued and the applicability of Koopman methods is unclear.

\item Nelson's mechanics provides a stochastic differential equation which allows us to track the evolution of the quantum probability $\rho$ for a non-zero solution of the Schrödinger equation in real time~\eqref{eq:TDSE}, which also includes the stationary equation~\eqref{eq:TISE} by means of~\eqref{eq:TDSE solution}. Koopman operator methods can then be used to analyze these dynamics further, as shown in Section~\ref{subsec:cca_example}, however, additional solutions to the Schrödinger equation cannot be obtained directly from such models.

\item Given a positive, real-valued solution of~\eqref{eq:TISE}, if we construct a time-dependent solution $\psi$ by~\eqref{eq:TDSE solution}, then we clearly have $\nabla S \equiv 0$ in the decomposition $\psi = e^{R + iS}$. Thereby we see that Nelson's mechanics and the ground-state transformation lead to the same system in this particular case.
\end{enumerate}

Finally, we briefly discuss the applicability of Nelson's stochastic differential equation (or the ground state system) in the stationary case, if the wave function is required to be anti-symmetric and therefore possesses zeros by necessity. Let $ \mathcal{N} $ be the nodal set of the ground state $ \psi $, i.e., $ \psi(x) = 0 \quad \forall x\in \mathcal{N} $, and let $\mathcal{N}^c$ be the complement of the nodal set. If the tiling property holds \cite{Ceperley1991}, then by anti-symmetry, we can write $ \psi =\pm e^{R} $, where the function $R = \log \abs{\psi}$ is the same within each connected component of $\mathcal{N}^c$.  Up to permutation, we therefore obtain the same stationary and ergodic stochastic differential equation within each component. The nodal set acts as a barrier between each of these dynamics, as $-R(x) \rightarrow \infty$ as $x$ approaches the nodal set, and it can be shown that the process in each component has zero probability of entering the nodal set \cite{Nelson66, Bacciagaluppi99}. Moreover, eigenfunctions of the generator will be the same up to permutations. By anti-symmetrizing appropriately, additional solutions of the stationary Schrödinger equation can then be obtained from analyzing Nelson's stochastic differential equation within any connected component of~$\mathcal{N}^c$. Data-driven methods that take symmetry and antisymmetry constraints into account have been recently proposed in \cite{KGNN21}.

\section{Optimal control formulation for the Schrödinger equation}
\label{sec:Optimal control formulation}

The stochastic formulations we have previously considered require the knowledge of at least one specific solution of the time-dependent (or time-independent) Schrödinger equation, either in real or imaginary time. In what follows, we will consider an optimal control problem that will actually allow us to compute a solution in imaginary time, requiring only knowledge of the potential energy and the initial condition.

\subsection{Hamilton--Jacobi--Bellman equation and Feynman--Kac formula}
\label{subsec:optimal_control}

We assume the state space is $\mathbb{X} = \mathbb{R}^d$, while noting that extensions to more general subsets of $\mathbb{R}^d$ are possible. First, we observe that for a positive solution $\psi = \psi(x, \tau)$ of the imaginary-time Schrödinger equation on $\mathbb{R}^d$, satisfying the initial condition $\psi(x, 0) = \psi_0(x)$, its negative logarithm satisfies a nonlinear partial differential equation (see also~\cite{Fleming2006}[Ch.\ VI.3] and Appendix~\ref{app:Continuity equation}). In this section, the symbol $\| \cdot \|$ always denotes the Euclidean vector norm.

\begin{lemma} \label{lem:log_transform_schroedinger}
Assume $\psi(x, \tau) = e^{R(x, \tau)} > 0$ solves~\eqref{eq:it TDSE}, and define $V = -R $. Then $V$ solves the nonlinear partial differential equation
\begin{equation*}
    \pd{V}{\tau} = \frac{1}{2}\Delta V - \frac{1}{2} \|\nabla V\|^2 + W.
\end{equation*}
\end{lemma}

\begin{proof}
Using that $-\frac{1}{2}\Delta \psi =  \psi\left[\frac{1}{2}\Delta V - \frac{1}{2} \|\nabla V\|^2 \right]$, we find:
\begin{equation*}
    0 = \pd{\psi}{\tau} + \mathcal{H}\psi = \psi\left[-\pd{V}{\tau} + \frac{1}{2}\Delta V - \frac{1}{2} \|\nabla V\|^2 + W\right].
\end{equation*}
Dividing by $\psi$, we obtain the result.
\end{proof}

The partial differential equation in Lemma~\ref{lem:log_transform_schroedinger} is well-known as the \emph{Hamilton--Jacobi--Bellman equation} (HJB equation) in stochastic optimal control. It is associated with the following control problem:
\begin{equation} \label{eq:optimal_control_problem}
\begin{split}
J(x, \tau) = \min_{\substack{u \colon [\tau, T] \to \mathbb{R}^d \\ u \in \mathcal{U}}} \, & \mathbb{E}^x \left[ \int_\tau^{T} W(X_s^u) + \frac{1}{2} \|u(s) \|^2 \, \mathrm{ds} - \log(\psi_0(X_T^u)) \right] \\
\text{s.t.} \quad \mathrm{d}X_s^u &= u(s) \, \mathrm{d}s + \mathrm{d}B_s, \quad X_\tau^u = x.
\end{split}
\end{equation}
Here, $\mathcal{U}$ is a given class of admissible controls $u$, see~\cite{Fleming2006}[Ch. IV.2] for the technical details. In a control context, the potential $W$ would be called \emph{running cost}, while $-\log(\psi_0)$ is called \emph{terminal cost}. The function $J$, finally, is called \emph{value function} of the optimal control problem.

If the HJB equation possesses a strong solution, then this solution in fact equals the value function $J$, and the wave function can be inferred from the control problem~\eqref{eq:optimal_control_problem}:

\begin{theorem}[\cite{Fleming2006}, Theorem VI.4.1]
\label{thm:bellman_eq}
Let $V \in C^{2,1}(\mathbb{R}^d, [0, T])$ be a strong solution to the HJB equation
\begin{align}
\label{eq:hjb_pde}
- \pd{V(x, \tau)}{\tau} &= \frac{1}{2} \Delta V(x, \tau) + W(x) - \frac{1}{2} \|\nabla V(x, \tau)\|^2, \\
\nonumber V(x, T) &= -\log(\psi_0(x)),
\end{align}
satisfying a linear growth condition $\|\nabla V\| \leq C(1 + \|x\|)$. Then $V = J$ for the optimal control problem~\eqref{eq:optimal_control_problem}, and it follows that $\psi(x, \tau) = \exp(-J(x, T - \tau))$ solves the Schrödinger equation~\eqref{eq:it TDSE} with initial condition $\psi(x, 0) = \psi_0(x)$. The optimal Markov control policy is given by $u^*(x, \tau) = \nabla J(x, \tau)$. 
\end{theorem}

The key to solving the control problem~\eqref{eq:optimal_control_problem} is to estimate expectations of specific observables at all time instances in the control interval $[\tau, T]$, which can then be integrated due to linearity of the expectation. This can be achieved efficiently using Koopman-based methods, as we will explain in Section~\ref{subsec:solution_control_problem}. Let us also spell out the connection between the control formulation and the transformation in Lemma~\ref{lem:Transformation TDSE}. As $\psi(x, \tau) = \exp(-J(x, T - \tau))$ obtained from the value function is a solution of the Schrödinger equation in imaginary time, we can directly apply Lemma~\ref{lem:Transformation TDSE} to $\psi$. The drift of the resulting stochastic differential equation is
\begin{equation*}
    -\frac{\nabla \psi}{\psi}(x, \tau) = - \nabla \log(\psi)(x, \tau) = \nabla J(x, T - \tau) = u^*(x, T - \tau).
\end{equation*}
The transformation thus leads to the generator of the optimally controlled stochastic differential equation for the problem~\eqref{eq:optimal_control_problem}.

\begin{remark} 
From the Girsanov theorem~\cite{Oksendal2013} and Jensen's inequality, it follows that the value function $J$ always overestimates the negative logarithm $V$ of the wave function $\psi$ (see \cite{Fleming2006}[Ch. VI]):
\begin{equation} \label{eq:cert_equiv_expectation}
\begin{split}
J(x, \tau) &= \min_{u \in \mathcal{U}} \mathbb{E}^x \left[ \int_\tau^{T} W(X_s^u) + \frac{1}{2} \|u(s) \|^2 \, \mathrm{ds} - \log(\psi_0(X_T^u)) \right] \\
&\geq -\log\left(\mathbb{E}^{x}\left[ \exp\left[ - \int_\tau^T W(B_s^x)\,\mathrm{d}s\right] \ts \psi_0(B_T^x)\right]\right) = -\log(\psi(x, T - \tau)) \\
 &= V(x, T - \tau).
\end{split}
\end{equation}
The second equality in~\eqref{eq:cert_equiv_expectation} is known as the \emph{Feynman--Kac formula} \cite{Feynman1948,Kac1949}, highlighting yet another stochastic interpretation of Schrödinger's equation. However, as the expectation in the Feynman--Kac formula acts on a nonlinear function of the time integral, Koopman methods are not directly applicable. Therefore, we prefer the control formulation~\eqref{eq:optimal_control_problem}.
\end{remark}

\subsection{Solution of the control problem}
\label{subsec:solution_control_problem}

A wealth of numerical methods for the solution of the control problem~\eqref{eq:optimal_control_problem} exist, for instance using dynamic programming \cite{BS96,FR12} or Monte Carlo sampling \cite{Rod07}. More recently, reformulations as deterministic control problems via the Koopman generator, which allow for a significant reduction of the complexity, were proposed \cite{KNPNCS20, POR20, NPP+21}. For an application to quantum control, see \cite{GKD+21}.
In particular for control-affine systems, i.e.,
\begin{equation} \label{eq:ControlSDE}
    \mathrm{d}X_s^u = \left(b(X_s^u) + G(X_s^u)u(s)\right) \ts \mathrm{d}s + \sigma(X_s^u) \ts \mathrm{d}B_s,
\end{equation}
and any observable function $\phi$, the linearity of the Koopman generator allows us to construct a deterministic bilinear surrogate model for the observed quantity $z(s) = \mathbb{E}^x[\phi(X_s^u)]$ using a finite set of Koopman generators \cite{POR20}.
Assuming that $\mathrm{dim}(u)=d$, we introduce the finite set $U=\{0, e_1, \ldots, e_d\}$, where the $e_j$ are the standard Euclidean basis vectors. By fixing the control in \eqref{eq:ControlSDE} to the elements of $U$, we obtain $d+1$ autonomous systems as well as $d+1$ associated Koopman generators
\begin{equation*}
\mathcal{L}_{U_j} f = \left(b + G \ts U_j \right)\cdot \nabla f + \frac{1}{2} a : \nabla^2 f, \qquad j=0,\ldots,d.
\end{equation*}
Introducing
\begin{equation*}
    \mathcal{A} = \mathcal{L}_{0} \quad\text{and}\quad \mathcal{B}_j = \mathcal{L}_{U_j} - \mathcal{L}_{0},
\end{equation*}
for $ j = 1, \dots, d $, we obtain a bilinear system for $z(s) = \mathbb{E}^x[\phi(X_s^u)]$:
\begin{align}\label{eq:ODE_Koopman}
\dot{z} = \mathcal{A}z + \left(\sum_{j=1}^d \mathcal{B}_j u_j\right)z,
\end{align}
cf.\ \cite{POR20} for a detailed derivation.

Returning to the imaginary-time Schrödinger equation, since the stochastic differential equation in \eqref{eq:optimal_control_problem} is control-affine, i.e., of the form \eqref{eq:ControlSDE}, we can in principle train a Koopman-based surrogate model \eqref{eq:ODE_Koopman} using data. However, one issue that we immediately face is that the system in \eqref{eq:optimal_control_problem} is unstable for $u\neq 0$, which makes accurate surrogate modeling extremely challenging.
To avoid this issue, we introduce a new control variable $\nu$ (with $\nu(s)\in\R^{d_u}$) and a state-dependent matrix $G(x)\in\R^{d \times d_u}$ satisfying $\mathrm{rank}(G(x)) = d$ for all $x \in \mathbb{R}^d$, and set
\begin{equation*}
    u(s) = G(X_s^u) \ts \nu(s).
\end{equation*}
This way, the dynamic constraint in \eqref{eq:optimal_control_problem} is replaced by
\begin{equation} \label{eq:optimal_control_problem_3}
    \mathrm{d}X_s^u = G(X_s^u) \ts \nu(s) \ts \mathrm{d}s + \mathrm{d}B_s,
\end{equation}
and we may choose $G(x)$ in such a way that the system corresponding to each $U_j$ becomes stable, rendering data-driven approximations feasible. Note that the Bellman equation for the stabilized control problem is still given by~\eqref{eq:hjb_pde}, hence the connection to the Schrödinger equation remains valid, see Appendix~\ref{app:stabilized_control}.

As a final step for the transformation of \eqref{eq:optimal_control_problem} to a deterministic system, we need to ensure that all individual terms in the objective function in \eqref{eq:optimal_control_problem} can be computed by a finite-dimensional model for the Koopman generator. Since, by linearity, we have
\begin{align}
&\mathbb{E}^x \left[ \int_\tau^{T} W(X_s^u) + \frac{1}{2}\|G(X_s^u)\nu(s) \|^2 \, \mathrm{ds} - \log(\psi_0(X_T^u)) \right] \notag\\
= &\int_\tau^{T} \left\{\mathbb{E}^x \left[W(X_s^u)\right] + \frac{1}{2}\mathbb{E}^x\left[(\nu(s))^\top (G(X_s^u))^\top G(X_s^u)\ts \nu(s)\right] \right\} \, \mathrm{ds} - \mathbb{E}^x\left[ \log(\psi_0(X_T^u)) \right] \notag\\
= &\int_\tau^{T} \left\{ \mathbb{E}^x \left[W(X_s^u)\right] + \frac{1}{2}\left(\sum_{i,j}\nu_i(s) \mathbb{E}^x\left[\left((G(X_s^u))^\top G(X_s^u)\right)_{i,j}\right]\nu_j(s)\right) \right\} \, \mathrm{ds} \label{eq:transformation_J}\\ &- \mathbb{E}^x\left[ \log(\psi_0(X_T^u)) \right], \notag
\end{align}
a straightforward solution is to directly include these terms in the dictionary $ \Phi = (\phi_1, \dots, \phi_n) $ used for learning the Koopman model (cf. Section~\ref{subsec:koopman_operator}), i.e.,
\begin{equation*}
    z(s) = \begin{bmatrix} W(X_s^u) & \log(\psi_0(X_s^u)) & \left((G(X_s^u))^\top G(X_s^u)\right)_{i,j} & \ldots\end{bmatrix},
\end{equation*}
where the third term represents all the entries of $G^\top G$.
\begin{remark}
We will see in the numerical results in Section \ref{sec:numerics_control} that $G$ can be chosen such that \eqref{eq:optimal_control_problem_3} becomes an Ornstein--Uhlenbeck process for each fixed control in $ U $. Depending on the specific problem setup, the dictionary entries then mainly consist of monomials up to a certain order.
\end{remark}

\section{Data-driven methods for quantum systems}
\label{sec:Data-driven methods for quantum systems}

A plethora of data-driven methods for the approximation of the Koopman operator and Perron--Frobenius operator from simulation or measurement data have been developed over the last years
\cite{BMM12, WKR15, KKS16}, including many tensor \cite{NSVN15, KGPS18, NGKC21}, kernel \cite{SP15, WRK15, KSM19, Tian20}, or neural network \cite{LDBK17, Otto19, MPWN18} extensions aiming at mitigating the curse of dimensionality. We will show how such methods (or generalizations thereof) can be used to analyze quantum systems. These data-driven methods can either be applied directly to data obtained by solving the Schrödinger equation or the corresponding stochastic formulations. We will highlight only a few potential use cases by first showing how the methods are typically used for classical systems and then applying these methods to quantum systems.

\subsection{Dynamic mode decomposition}
\label{sec:DMD}

One of the simplest but also most popular methods to estimate the Koopman operator from data is \emph{dynamic mode decomposition} (DMD) \cite{Schmid10, TRLBK14}. 

\paragraph{Conventional DMD.}

Let $ \Delta_t $ be a fixed lag-time and $ \Theta^{\Delta_t} $ the flow map associated with an arbitrary dynamical system of the form \eqref{eq:SDE}. Given training data $ \{ (x^{(i)}, y^{(i)}) \}_{i=1}^m $, where $ y^{(i)} = \Theta^{\Delta_t}(x^{(i)}) $, we define the data matrices $ X, Y \in \R^{d \times m} $ by
\begin{equation*}
    X =
    \begin{bmatrix}
        x^{(1)} & x^{(2)} & \dots & x^{(m)}
    \end{bmatrix}
    ~~ \text{and} ~~
    Y =
    \begin{bmatrix}
        y^{(1)} & y^{(2)} & \dots & y^{(m)}
    \end{bmatrix}.
\end{equation*}
DMD is based on the assumption that a linear relationship between the inputs and outputs exists, i.e., $ y^{(i)} = A \ts x^{(i)} $. The matrix $ A \in \R^{d \times d} $ is then estimated by solving the regression problem
\begin{equation*}
    \min_{A \in \R^{d \times d}} \norm{Y - A \ts X}_F.
\end{equation*}
The solution is given by $ A = Y \ts X^+ $, where $ ^+ $ denotes the pseudoinverse. The DMD eigenvalues and modes are defined to be the eigenvalues and eigenvectors of $ A $. If $ d \gg m $, then estimating the full matrix $ A $ is numerically inefficient, even storing $ A $ might be infeasible. There are many different algorithms to compute the DMD eigenvalues and modes without computing the full matrix $ A $, see, for example, \cite{TRLBK14}.

\paragraph{DMD for quantum systems.}

Assuming the Hamiltonian $ \mathcal{H} $ is time-independent, the formal solution of \eqref{eq:TDSE} can be written as
\begin{equation*}
    \psi(x, t + \Delta_t) = e^{-i \ts \Delta_t \mathcal{H}} \ts \psi(x, t).
\end{equation*}
By discretizing the spatial domain of a quantum system and replacing the second-order derivatives by finite-difference approximations, the partial differential equation reduces to a system of ordinary differential equations. Let $ x_r \in \R^{d_r} $ denote the vector of grid points, then we can generate training data $ \big\{ (\psi_0^{(i)}(x_r), \psi_{\Delta_t}^{(i)}(x_r)) \big\}_{i=1}^m $, where $ \psi_0^{(i)}(x_r) $ is an initial condition and $ \psi_{\Delta_t}^{(i)}(x_r) $ the corresponding solution at time $ \Delta_t $. This can be written in matrix form as
\begin{equation*}
    \Psi_0 =
    \begin{bmatrix}
        \psi_0^{(1)}(x_r) & \psi_0^{(2)}(x_r) & \dots & \psi_0^{(m)}(x_r)
    \end{bmatrix}
    ~\text{and}~
    \Psi_{\Delta_t} =
    \begin{bmatrix}
        \psi_{\Delta_t}^{(1)}(x_r) & \psi_{\Delta_t}^{(2)}(x_r) & \dots & \psi_{\Delta_t}^{(m)}(x_r)
    \end{bmatrix}.
\end{equation*}
As in the case of standard DMD, we can determine a matrix that approximates the dynamics of the system, i.e.,
\begin{equation*}
    A = \Psi_{\Delta_t} \ts \Psi_0^+.
\end{equation*}
That is, the matrix $ A $ propagates discretized solutions of the Schrödinger equation in time. The eigenvalues $ \mu_\ell $ and eigenvectors $ v_\ell $ of the matrix $ A $ can then be used to approximate eigenvalues and eigenfunctions of the Schrödinger equation. In order to obtain eigenvalues $ \lambda_\ell $ of the Schrödinger equation, we define
\begin{equation*}
    \lambda_\ell = \frac{i}{\Delta_t} \log(\mu_\ell).
\end{equation*}

\paragraph{Numerical results.}

Let us consider the quantum harmonic oscillator. We choose the lag time $ \Delta_t = 0.1 $ and discretize the interval $ [-5, 5] $ using $ d = 100 $ equidistant grid points, which we denote by $ x_r $. This spatial discretization turns the partial differential equation \eqref{eq:TDSE} into a system of coupled complex-valued ordinary differential equations. We then generate $ m = 200 $ initial conditions $ \psi_0^{(i)}(x_r) = \mathds{1}_{I^{(i)}}(x_r) \in \mathbb{C}^{d_r} $, where $ \mathds{1}_{I^{(i)}} $ denotes an indicator function for a randomly chosen interval $ I^{(i)} \subseteq [-5, 5] $. The corresponding vectors $ \psi_{\Delta_t}^{(i)}(x_r) \in \mathbb{C}^{d_r} $ are computed by solving the resulting system of ordinary differential equations using a standard Runge--Kutta integrator. (Note that we are not using the stochastic dynamics derived in Section \ref{sec:Stochastic descriptions} in this case.) With the aid of the data matrices $ \Psi_0, \Psi_{\Delta_t} \in \mathbb{C}^{d_r \times m} $, we compute the matrix $ A $ and its eigenvalues and eigenvectors. The results are shown in Figures~\ref{fig:QHO_DMD}(a) and~\ref{fig:QHO_DMD}(b). The highlighted eigenvalues are listed in Table~\ref{tab:DMD eigenvalues}. The eigenvalues and eigenvectors are good approximations of the analytically computed results, see Example~\ref{ex:1D problems}.

Instead of solving the Schrödinger equation \eqref{eq:TDSE}, we can also solve the imaginary-time Schrödinger equation \eqref{eq:it TDSE} and apply DMD. The eigenvalues of the Schrödinger equation can then be estimated via
\begin{equation*}
    \lambda_\ell = -\frac{1}{\Delta_t}\log(\mu_\ell).
\end{equation*}
The results are shown in Figures~\ref{fig:QHO_DMD}(b) and \ref{fig:QHO_DMD}(c) and the corresponding eigenvalues are also listed in Table~\ref{tab:DMD eigenvalues}. We again obtain accurate estimates of the eigenvalues and eigenfunctions. The advantage of the imaginary-time formulation is that the eigenvalues are well ordered and not distributed on the complex unit circle and consequently only determined up to $ 2 \ts \pi $.
\begin{figure}
    \centering
    \begin{minipage}[t]{0.32\textwidth}
        \centering
        \subfiguretitle{(a)}
        \includegraphics[height=0.22\textheight]{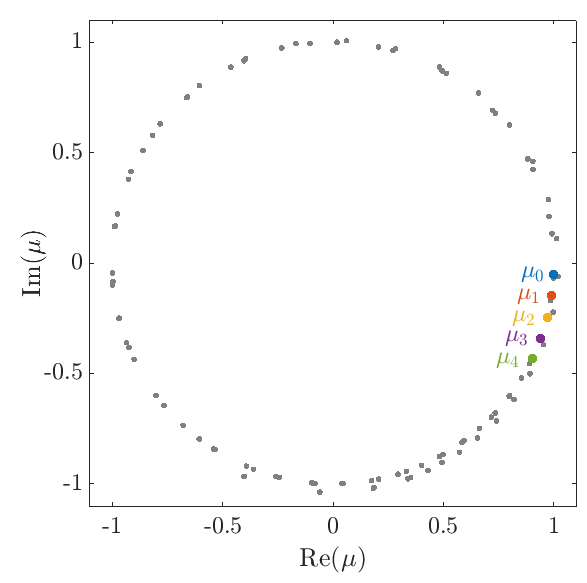}
    \end{minipage}
    \begin{minipage}[t]{0.32\textwidth}
        \centering
        \subfiguretitle{(b)}
        \includegraphics[height=0.22\textheight]{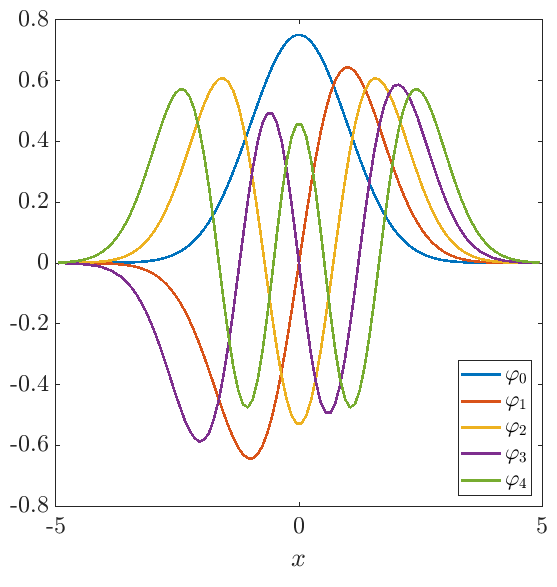}
    \end{minipage}
    \begin{minipage}[t]{0.32\textwidth}
        \centering
        \subfiguretitle{(c)}
        \includegraphics[height=0.22\textheight]{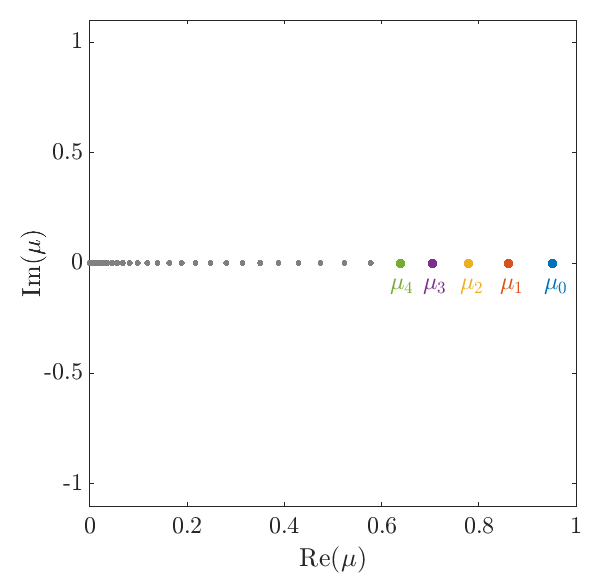}
    \end{minipage}
    \caption{(a) Eigenvalues of the matrix $ A $ associated with the Schrödinger equation. (b)~A few select eigenvectors corresponding to the highlighted eigenvalues with the same color. (c) Eigenvalues using the imaginary-time Schrödinger equation.}
    \label{fig:QHO_DMD}
\end{figure}
\begin{table}
    \centering
    \caption{Eigenvalues computed by applying DMD to the Schrödinger equation \eqref{eq:TDSE} and the imaginary-time Schrödinger equation \eqref{eq:it TDSE}.}
    \label{tab:DMD eigenvalues}
    \begin{tabular}{c|C{2.4cm}C{1.1cm}|C{1.1cm}C{1.1cm}}
          & \multicolumn{2}{c|}{real time} & \multicolumn{2}{c}{imaginary time} \\ \hline
        $ \ell $ & $ \mu_\ell $ & $ \lambda_\ell $ & $ \mu_\ell $ & $ \lambda_\ell $ \\ \hline
        0 & $0.999 - 0.045 \ts i$ & 0.499 & 0.951 & 0.450 \\
        1 & $0.989 - 0.149 \ts i$ & 1.498 & 0.861 & 1.498 \\
        2 & $0.969 - 0.247 \ts i$ & 2.496 & 0.779 & 2.496 \\
        3 & $0.940 - 0.342 \ts i$ & 3.492 & 0.705 & 3.492 \\
        4 & $0.901 - 0.434 \ts i$ & 4.487 & 0.638 & 4.487
    \end{tabular}
\end{table}
The matrix $ A $ could now also be used to predict the evolution of the system. This is another important use case of DMD.

\subsection{Extended dynamic mode decomposition}

\emph{Extended dynamic mode decomposition} (EDMD) \cite{WKR15, KKS16} can be regarded as a nonlinear variant of DMD.

\paragraph{Conventional EDMD.}

The data is first embedded into a typically higher-dimensional feature space using a nonlinear transformation $ \phi \colon \R^d \to \R^n $. The chosen basis functions $ \phi_1, \dots, \phi_n \colon \R^d \to \R $ could, for instance, be indicator functions, monomials, or radial basis functions, the optimal choice depends on the system for which we aim to approximate the Koopman operator. We define the transformed data matrices $ \Phi_x, \Phi_y \in \R^{n \times m} $ by
\begin{equation*}
    \Phi_x = \begin{bmatrix} \phi(x^{(1)}) & \phi(x^{(2)}) & \dots & \phi(x^{(m)}) \end{bmatrix}
    ~\text{and}~
    \Phi_y = \begin{bmatrix} \phi(y^{(1)}) & \phi(y^{(2)}) & \dots & \phi(y^{(m)}) \end{bmatrix}.
\end{equation*}
The minimization problem thus becomes
\begin{equation*}
    \min_{K \in \R^{n \times n}} \big\| \Phi_y - K^\top \Phi_x \|_F.
\end{equation*}
The solution is now given by
\begin{equation*}
    K^\top = \Phi_y \ts \Phi_x^+ = \big(\Phi_y \ts \Phi_x^\top\big) \big(\Phi_x \ts \Phi_x^\top\big)^+ = C_{yx} \ts C_{xx}^+.
\end{equation*}
The matrices $ C_{xx} $ and $ C_{xy} = C_{yx}^\top $ are empirical estimates of the matrices $ C^t $ and $ A^t $ introduced in Section~\ref{subsec:koopman_operator}. The matrix $ K $ is the representation of the Koopman operator projected onto the space spanned by the basis functions $ \phi $. Approximate eigenvalues and eigenfunctions of the Koopman operator can be computed by determining the eigenvalues and eigenvectors of $ K $. If the number of basis functions $ n $ is larger than the number of snaphots $ m $, a dual method called \emph{kernel EDMD} \cite{WRK15, KSM19} can be used. The basis functions are then implicitly defined by the feature map associated with the kernel. Instead of an eigenvalue problem involving covariance and cross-covariance matrices of size $ n \times n $, a problem involving (time-lagged) Gram matrices of size $ m \times m $ needs to be solved. EDMD can be used in the same way to compute eigenfunctions of the Perron--Frobenius operator. The matrix representation in that case is $ P^\top = C_{xy} \ts C_{xx}^+ $. See \cite{KKS16} for a detailed derivation.

\paragraph{EDMD applied to quantum systems.}

We could also apply EDMD or kernel EDMD to the data generated in Section~\ref{sec:DMD}, but we will now present a different use case and show how EDMD and its extensions can be applied to the stochastic formulations derived in Section~\ref{sec:Stochastic descriptions}. We assume that the process is stationary, non-stationary problems will be discussed in the following subsection.

\paragraph{Numerical results.}

\begin{figure}
    \centering
    \begin{minipage}[t]{0.32\textwidth}
        \centering
        \subfiguretitle{(a)}
        \includegraphics[height=0.215\textheight]{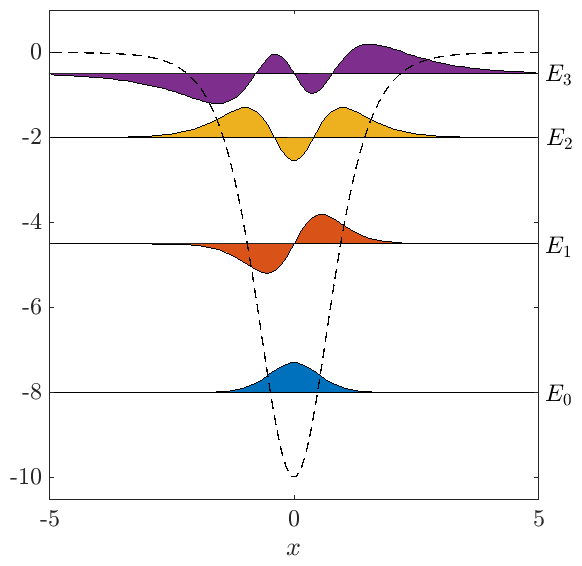}
    \end{minipage}
    \begin{minipage}[t]{0.32\textwidth}
        \centering
        \subfiguretitle{(b)}
        \includegraphics[height=0.215\textheight]{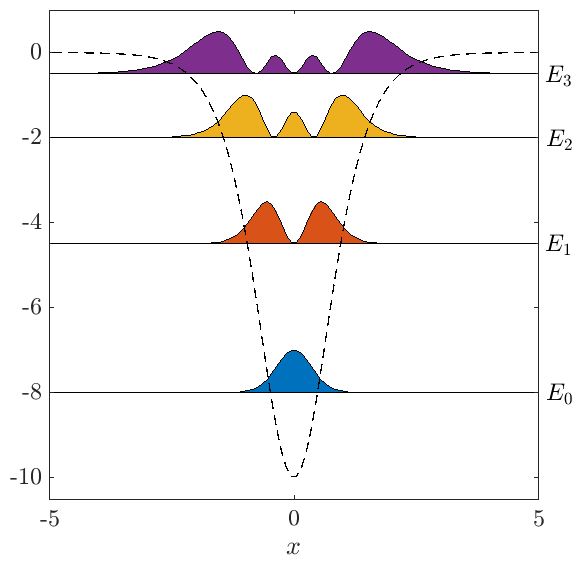}
    \end{minipage}
    \begin{minipage}[t]{0.32\textwidth}
        \centering
        \subfiguretitle{(c)}
        \includegraphics[height=0.215\textheight]{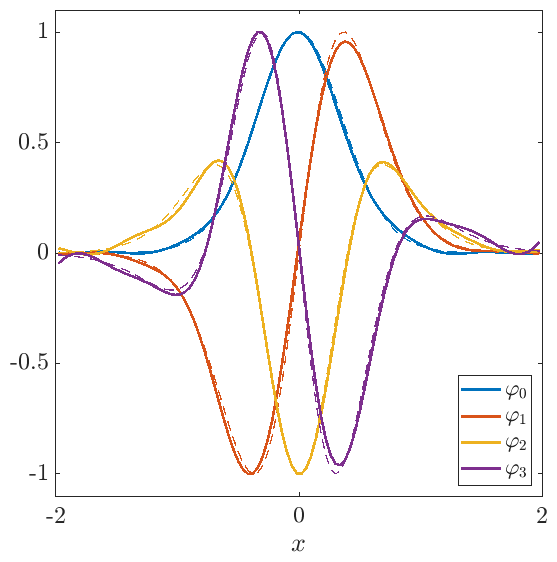}
    \end{minipage}
    \caption{(a) Eigenfunctions of the Schrödinger operator computed with the aid of gEDMD. (b)~Corresponding probability densities. The dashed line represents the Pöschl--Teller potential. (c)~Eigenfunctions of the Perron--Frobenius operator obtained by applying EDMD to time-series data. The dashed lines are the analytically computed eigenfunctions. Note that we choose different axis limits for the sake of visibility.}
    \label{fig:PoeschTeller}
\end{figure}

Let us consider the Pöschl--Teller potential introduced in Example~\ref{ex:1D problems} and set $ s = 4 $. The energy levels and (unnormalized) wave functions are then given by
\begin{align*}
    E_0 &= -8,           & \psi_0(x) &= \sech^4(x), \\
    E_1 &= -\frac{9}{2}, & \psi_1(x) &= \sech^3(x) \tanh(x), \\
    E_2 &= -2,           & \psi_2(x) &= \sech^2(x) \big(7 \tanh^2(x) - 1\big), \\
    E_3 &= -\frac{1}{2}, & \psi_3(x) &= \sech(x) \tanh(x) \big(7 \tanh^2(x) - 3\big),
\end{align*}
see also \cite{BH18:PoeschlTeller}. The eigenfunctions are shown in Figure~\ref{fig:PoeschTeller}. For the Schrödinger operator, we used kernel gEDMD---a variant of EDMD that directly approximates the infinitesimal generator of the Koopman operator \cite{KNPNCS20}, but can also be used to approximate the Schrödinger operator \cite{KNH20} and is related to quantum Monte Carlo methods---with 100 randomly generated test points in the interval $ [-5, 5] $ and a Gaussian kernel with bandwidth $ \varsigma = 0.3 $. The resulting eigenvalues and eigenfunctions are in perfect agreement with the analytically computed ones. In order to compute eigenfunctions of the Perron--Frobenius operator, we generate 10,000 trajectories (using randomly drawn initial conditions) by integrating the stochastic differential equation with the aid of the Euler--Maruyama method, where the step size is $ h = 10^{-3} $ and the lag time $ \Delta_t = 0.1 $. Applying EDMD with a dictionary comprising 100 Gaussian functions with bandwidth $ \varsigma = 0.5 $ results in the generator eigenvalues $ \lambda_0 = 0 $, $ \lambda_1 = -3.49 $, $ \lambda_2 = -5.90 $, and $ \lambda_3 = -7.61 $. Shifting these eigenvalues according to Lemma~\ref{lem:Ground-state transformation}, we obtain $ E_0 = -8 $, $ E_1 = -4.51 $, $ E_2 = -2.1 $, and $ E_3 = -0.39 $, which is close to the true solutions. Dividing the eigenfunctions of the Perron--Frobenius operator by the ground state, we obtain estimates of the higher-energy states of the quantum system.

Analogously, we could use eigenfunctions of the Koopman operator to compute higher-energy states. The only difference is that we then have to multiply the eigenfunctions of the Koopman operator by the ground state. This illustrates the close relationship between the Schrödinger operator, the Koopman operator, and the Perron--Frobenius operator. Numerical results for the quantum harmonic oscillator and the hydrogen atom can be found in \cite{KNH20}.

\subsection{Canonical correlation analysis}
\label{subsec:cca_example}

\emph{Canonical correlation analysis} (CCA) was originally developed to maximize the correlation between two multi-dimensional random variables \cite{Hotelling36}. It was shown in \cite{KHMN19, WuNo20} that CCA, when applied to Lagrangian data, can be interpreted as a composition of the Koopman operator and a reweighted Perron--Frobenius operator.

\paragraph{Conventional CCA.}

So far, we assumed that the dynamical system is time-homogeneous, i.e., the Koopman operator depends only on the lag time $ \Delta_t $. If the dynamics change over time, instead of computing eigenfunctions of the Koopman operator, typically eigenfunctions of a related forward-backward operator are computed. This leads to the notion of \emph{coherent sets} \cite{FrSaMo10, Froyland13}, which can be regarded as generalizations of metastable sets. Coherent sets are regions of the phase space that disperse slowly, i.e., particles are (almost) trapped in these sets. We again collect data as before. The matrix of the operator representing the forward-backward dynamics is then given by
\begin{equation*}
    L^\top = C_{xy} \big(C_{yy} + \varepsilon I\big)^{-1} C_{yx} \big(C_{xx} + \varepsilon I\big)^{-1},
\end{equation*}
where $ \varepsilon $ is a regularization parameter that ensures that the inverse exists. Alternatively, we could use the pseudoinverse. In what follows, the eigenvalues of this matrix are denoted by $ \kappa_\ell $. A detailed derivation can be found in \cite{KHMN19}. Similar kernel-based variants have been proposed in \cite{Tian20} and a deep-learning counterpart in \cite{MPWN18}.

\paragraph{CCA applied to quantum systems.}

We have seen that for the non-stationary case the ground-state transformation and Nelson's formulation result in different models. We now apply CCA to data generated by Nelson's stochastic mechanics.

Assume the harmonic oscillator is in the coherent state
\begin{equation*}
    \psi_c(x, t) = \left(\tfrac{\omega}{\pi}\right)^{1/4} e^{-\tfrac{\omega}{2} \big(x-x_0 \ts \cos(\omega \ts t)\big)^2 - \tfrac{1}{2} \ts i \ts \omega \ts t - i \ts \omega \big(x \ts x_0 \sin(\omega \ts t) - \tfrac{1}{4} x_0^2 \sin(2 \ts \omega \ts t)\big)},
\end{equation*}
see \cite{Grabert79}. The corresponding probability density is a wave whose center is periodically moving from $ x_0 $ to $ -x_0 $. If follows that
\begin{equation*}
    \nabla R = -\omega \big(x - x_0 \cos(\omega \ts t)\big)
    \quad \text{and} \quad
    \nabla S = -\omega \big(x_0 \sin(\omega \ts t)\big)
\end{equation*}
so that
\begin{equation*}
    b(x, t) = -\omega\big(x - x_0 \cos(\omega \ts t) + x_0 \sin(\omega \ts t)\big) = -\omega \left(x + \sqrt{2} \ts x_0 \sin(\omega \ts t -\tfrac{\pi}{4})\right),
\end{equation*}
which can be viewed as an Ornstein--Uhlenbeck process with time-dependent but periodic shift. We will now use a superposition of this coherent state and the second eigenfunction, i.e., $ \psi = \psi_2 + \frac{1}{2} \ts \psi_c $ (properly normalized) with $ x_0 = 2 $, to construct a system with nontrivial coherent sets. The required osmotic and current velocities for such a superposition are derived in Appendix~\ref{app:Superposition}.

\paragraph{Numerical results.}

\begin{figure}
    \centering
    \begin{minipage}{0.49\textwidth}
        \centering
        \subfiguretitle{\hspace{2.1em}(a)}
        \includegraphics[width=0.95\textwidth]{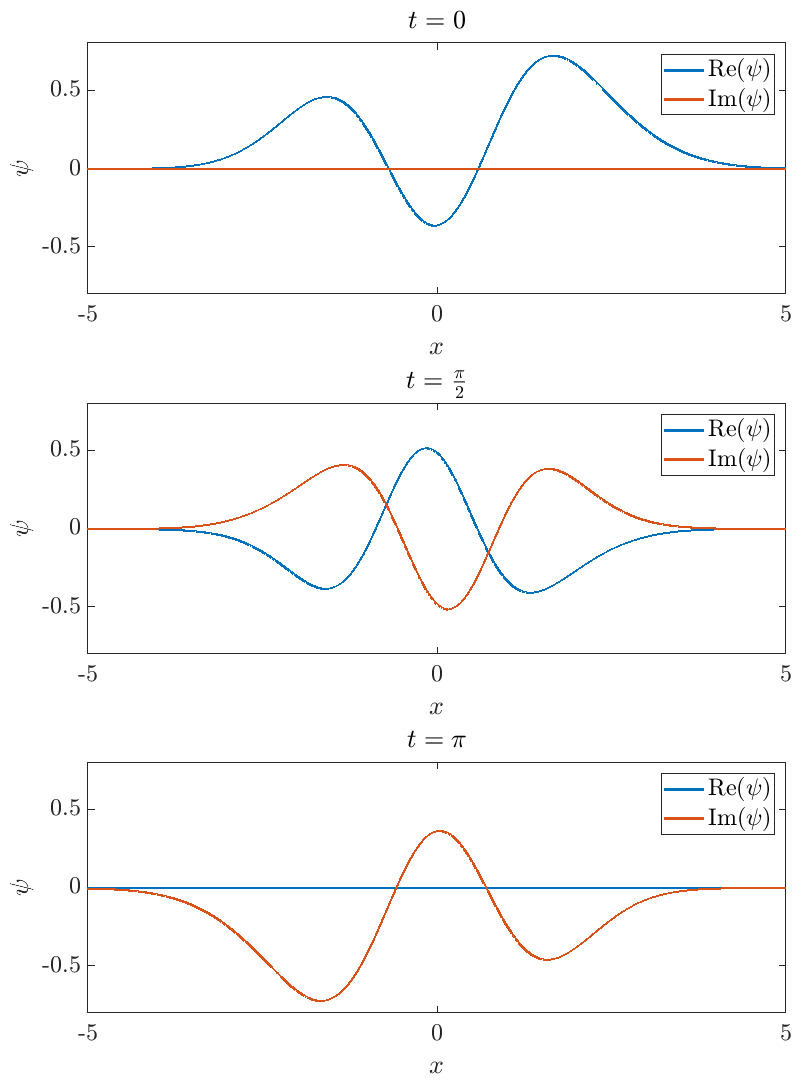}
    \end{minipage}
    \begin{minipage}{0.49\textwidth}
        \centering
        \subfiguretitle{\hspace{2.1em}(b)}
        \includegraphics[width=0.95\textwidth]{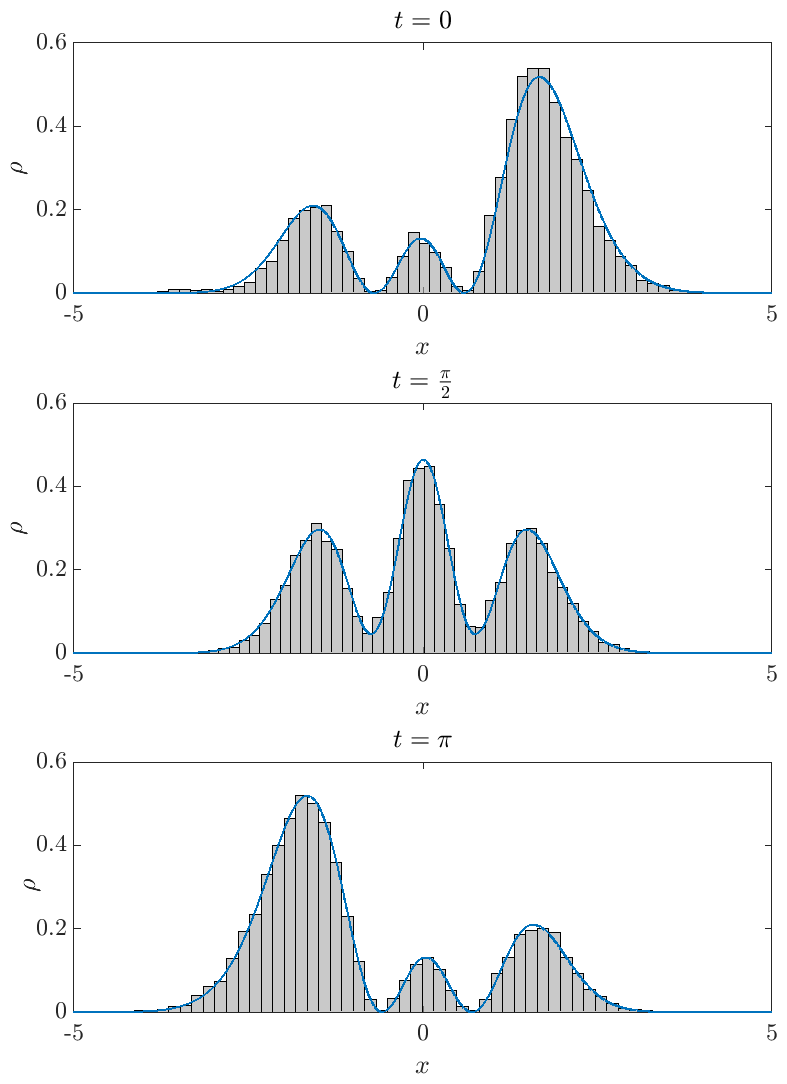}
    \end{minipage} \\[1ex]
    \begin{minipage}{0.49\textwidth}
        \centering
        \subfiguretitle{\hspace{2.1em}(c)}
        \includegraphics[width=0.95\textwidth]{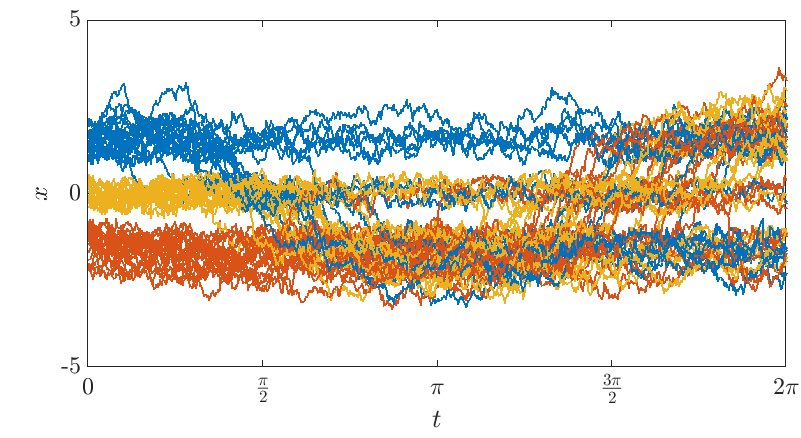}
    \end{minipage}
    \begin{minipage}{0.49\textwidth}
        \centering
        \subfiguretitle{\hspace{2.1em}(d)}
        \includegraphics[width=0.95\textwidth]{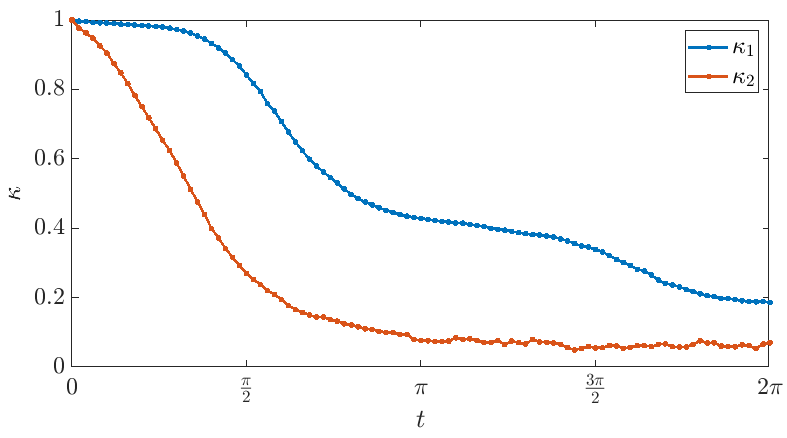}
    \end{minipage}
    \caption{(a) Wave function $ \psi $ at different time points. (b) Corresponding analytically computed probability densities $ \rho $ and histograms of particles. (c) Subset of trajectories colored with respect to the obtained clustering into coherent sets. (d) Coherence as a function of the time $ t $.}
    \label{fig:QHO}
\end{figure}

We generate 10,000 particles sampled from the initial distribution and integrate the stochastic differential equation from $ t = 0 $ to $ t = 2 \ts \pi $. At $ t = 0 $ the wave function has two nodes and there are thus three invariant sets. With increasing time, these sets become connected and particles can move to the neighboring sets. At $ t = \pi $ the sets are disconnected again. An illustration of the time-dependent wave function is shown in Figure~\ref{fig:QHO}(a), the corresponding analytically computed and empirical probability densities in Figure~\ref{fig:QHO}(b). The distribution of the particles perfectly follows $ \rho $. For the initial sampling, we used Metropolis--Hastings. By applying kernel CCA to the trajectory data and clustering the dominant eigenfunctions, we identify three finite-time coherent sets associated with the time-dependent potential wells. Particles starting in one coherent set will remain is this set with a high probability (compared to the probability that they transition to another set). This can be seen in Figure~\ref{fig:QHO}(c), where we colored some trajectories according to the computed coherent sets. The coherence, however, decreases over time so that after a while the particles seem to be well-mixed as shown in Figure~\ref{fig:QHO}(d). The eigenfunction associated with $ \kappa_1 $ is negative for the left two sets and positive for the right one (or vice versa), while $ \kappa_2 $ distinguishes between the coherent set in the middle and the other two. It can be seen that $ \kappa_2 $ quickly decreases and is not distinguishable from noise for large $ t $. This is consistent with the trajectory data: The yellow set is dispersed quickly, whereas the blue set remains coherent for a longer time.

The example demonstrates that it is possible to analyze the probability flow associated with time-dependent wave functions using Koopman operator theory and data-driven methods to estimate eigenfunctions of associated forward-backward operators.

\section{Data-driven analysis of quantum systems via control}
\label{sec:numerics_control}

In Section \ref{subsec:solution_control_problem}, we have seen how solutions of the imaginary-time Schrödinger equation are related to solutions of the control problem \eqref{eq:optimal_control_problem}.
We call this approach \emph{DISCo} (Data-driven solution of the Imaginary-time Schrödinger equation via Control), and we now study DISCo in detail for two systems, namely the quantum harmonic oscillator and the hydrogen atom.

As discussed above, the stochastic differential equation in \eqref{eq:optimal_control_problem} is unstable for $u\neq0$. Thus, we aim to choose the matrix $G$ in \eqref{eq:optimal_control_problem_3} in such a way that the system becomes stable, while at the same time being easy to approximate from data.
A straightforward choice in this case is the Ornstein--Uhlenbeck process, for which we set
\begin{equation*}
G(x) = \begin{bmatrix} \text{diag}(-x) & I\end{bmatrix} \in \mathbb{R}^{d \times 2d} \quad \mbox{and} \quad \hat\nu(s) = \begin{bmatrix}
1 & \cdots & 1 & (\nu(s))^\top
\end{bmatrix}^\top \in \mathbb{R}^{2d}.
\end{equation*}
This process is stable for all $\nu$ such that the Koopman generators associated with
\begin{equation*}
    \mathrm{d}X_s^u = -(X_s^u-U_j) \, \mathrm{d}s + \mathrm{d}B_s, \quad j=1,\ldots,d+1, \quad U =\begin{bmatrix} 0 & e_1 & \cdots & e_d\end{bmatrix},
\end{equation*}
can be reliably approximated from data. In this case, the objective function in \eqref{eq:optimal_control_problem} can be transformed accordingly (cf.\ the general formulation \eqref{eq:transformation_J}):
\begin{equation} \label{eq:Objective_Koopman}
\begin{aligned}
J(x,\tau)&=\mathbb{E}^x \left[ \int_\tau^{T} \left\{ W(X_s^u) + \frac{1}{2}\|u(s) \|^2 \right\} \, \mathrm{ds} - \log(\psi_0(X_T^u)) \right] \\
&= \int_\tau^{T} \left\{ \mathbb{E}^x \left[W(X_s^u)\right] + \frac{1}{2} \ts \mathbb{E}^x\left[(X_s^u)^\top X_s^u-2(X_s^u)^\top\nu(s)+\nu^\top(s)\ts\nu(s)\right] \right\} \, \mathrm{ds} \\ &\phantom{=}- \mathbb{E}^x\left[ \log(\psi_0(X_T^u)) \right] \\
&= \int_\tau^{T} \mathbb{E}^x \left[W(X_s^u)\right] + \frac{1}{2}\ts\mathbb{E}^x\left[(X_s^u)^\top X_s^u\right] - \ts\mathbb{E}^x \left[(X_s^u)\right]^\top\nu(s)\\ &\phantom{=} + \frac{1}{2} \nu^\top(s)\ts\nu(s) \, \mathrm{ds} - \mathbb{E}^x\left[ \log(\psi_0(X_T^u)) \right].
\end{aligned}
\end{equation}
Consequently, the dictionary $\Phi$ needs to contain at least the terms (for $1 \leq i,\,j \leq d$): 
\begin{equation}
\label{eq:dict_ou_process}
\Phi(x)=\begin{bmatrix} W(x) & \log(\psi_0(x)) & x_i & x_i\, x_j & \ldots\end{bmatrix}.
\end{equation}

Choosing the remaining terms in the dictionary is a critical step. The literature on Koopman operator methods contains a broad selection of different options, e.g., monomial bases, Gaussians~\cite{Noe2013,NKPMN14}, indicator functions~\cite{Prinz2011,SS13}, reproducing kernels~\cite{SP15, WRK15, KSM19, Tian20}, deep learning architectures~\cite{LDBK17, Otto19, MPWN18}, and many more. However, as we expect quantum systems to behave quite differently compared to molecular or fluid dynamics simulations, we leave the selection of the dictionary as a topic for future research.

\paragraph{Example 1: The quantum harmonic oscillator.}
We consider the one-dimensional quantum harmonic oscillator, and aim to solve the imaginary-time equation~\eqref{eq:it TDSE} with initial condition equal to the ground state. We then have 
\begin{equation*}
W(x)=\frac{x^2}{2} \qquad \mbox{and} \qquad \psi_0(x) = e^{-\frac{x^2}{2}} \quad \Rightarrow \quad \log(\psi_0(x)) = -\frac{x^2}{2}.
\end{equation*}
The analytical solution is $\psi(x, \tau) = e^{-\frac{x^2 + \tau}{2}}$. Thus, if we consider monomials up to order at least two in the dictionary $\Phi$, i.e., $\Phi(x) = \begin{bmatrix}
1 & x & x^2 & \ldots
\end{bmatrix}$, then \eqref{eq:Objective_Koopman} becomes
\begin{equation} \label{eq:ObjectiveQHO}
J(x,\tau) = \frac{1}{2} \int_\tau^{T} z_3(s) + z_3(s) - 2 \ts z_2(s) \ts \nu(s) + \nu^2(s) \ts \mathrm{ds} - \frac{1}{2}z_3(T).
\end{equation}
In order to obtain an approximate solution for the imaginary-time Schrödinger equation, we solve the following optimal control problem on a grid of initial conditions $x$ and $\tau$:
\begin{equation}\label{eq:Opt_QHO}
\min_{\nu \colon [\tau, T] \to \mathbb{R}^d} \eqref{eq:ObjectiveQHO} \qquad s.t.\quad \eqref{eq:ODE_Koopman}~\mbox{with initial condition}~z(\tau)=\Phi(x).
\end{equation}
For the numerical approximation of the Koopman generators, we consider monomials up to order three as observables, i.e., $\Phi(x) = \begin{bmatrix}1 & x & x^2 & x^3\end{bmatrix}$. We then collect 30,000 data points for each of the $d+1 = 2$ systems using i.i.d.\ sampling from the interval $X_0^{U_j} \in [-3,3]$, where $U_j \in \{0, 1\}$. Finally, we calculate finite-dimensional matrix approximations of the operators $\mathcal{A}$ and $\mathcal{B}$ in \eqref{eq:ODE_Koopman} via the gEDMD algorithm from \cite{KNPNCS20}. We thus obtain a four-dimensional system describing the dynamics of $z$ within the subspace spanned by the monomial functions. The prediction capability of the Koopman model for the terms $\mathbb{E}^x[X_s^u]$ and $\mathbb{E}^x[(X_s^u)^2]$ is shown in Figure \ref{fig:Control_QHO_prediction}, and we observe very good performance over a large time horizon and for complex control inputs.

\begin{figure}
	\centering
	\includegraphics[width=\textwidth]{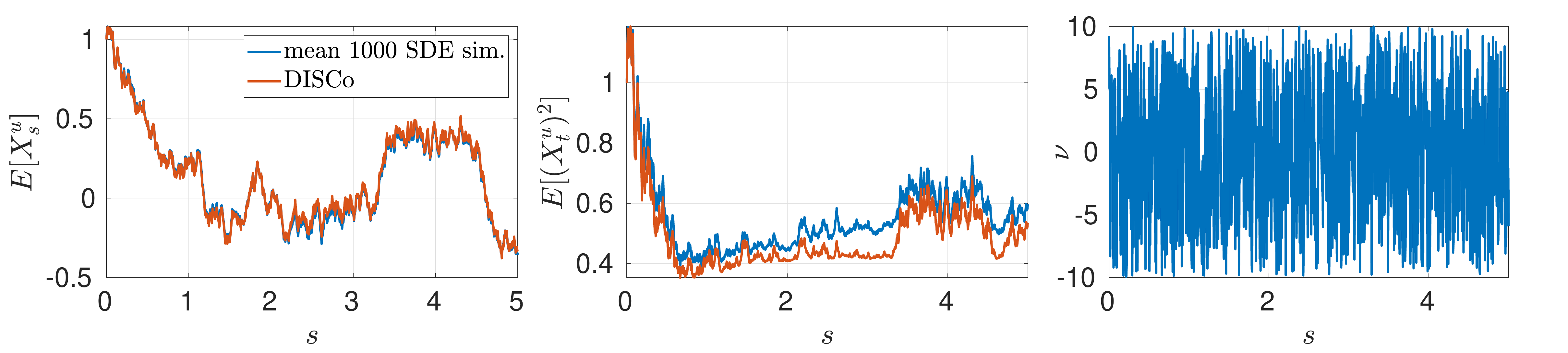}
	\caption{Prediction of $\mathbb{E}^x[X_s^u]$ and $\mathbb{E}^x\left[(X_s^u)^2\right]$, where $u(s)$ is sampled i.i.d.\ from the interval $[-10,10]$. The true expected values are calculated by averaging over 1000 repeated simulations of the SDE \eqref{eq:ControlSDE}.}
	\label{fig:Control_QHO_prediction}
\end{figure}

The resulting bilinear model can now be used to solve the deterministic optimal control problem \eqref{eq:Opt_QHO} instead of the stochastic problem \eqref{eq:optimal_control_problem}. The solution for the imaginary-time Schrödinger equation is then obtained using $\psi(x,\tau) = e^{-J(x,T-\tau)}$ (see Theorem~\ref{thm:bellman_eq}). Figure~\ref{fig:Control_QHO_optim} shows that---after scaling the solutions such that $\int \psi(x,0)\, \mathrm{dx}=1$---we observe a very high accuracy with an absolute error $\epsilon_{abs} = |\psi_\text{DISCo} - \psi_\text{analytical}| \approx 10^{-3}$ and relative error as shown in the right panel. More precisely, for $x \in [-2, 2]$, which covers approximately $95\%$ of the density, we have an average error of less than $0.4\%$, and a maximum error of $1.1\%$.

\begin{figure}
    \centering
    \subfiguretitle{(a) \hspace*{4.5cm} (b) \hspace*{4.5cm} (c)}
    \includegraphics[width=\textwidth]{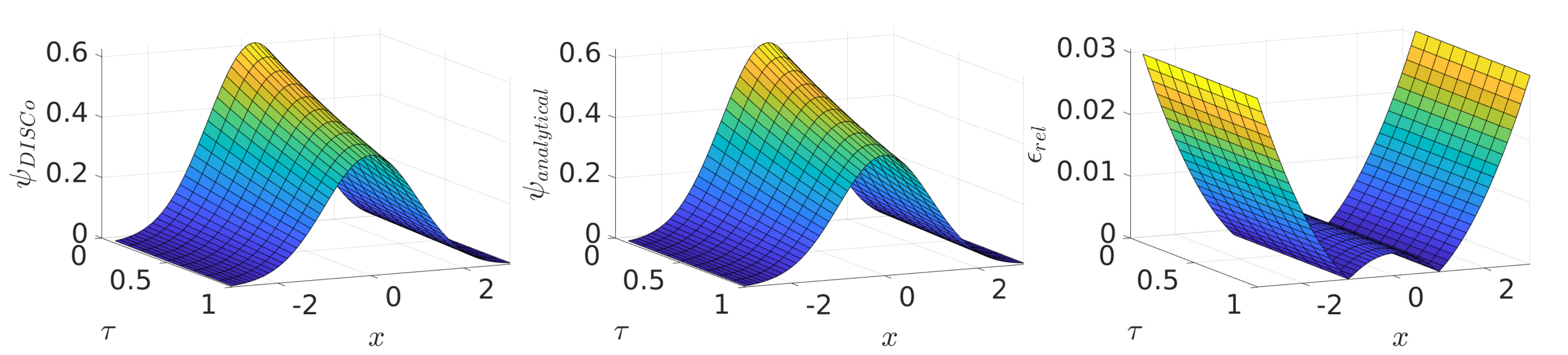}
    \caption{(a) Solution of the imaginary-time Schrödinger equation obtained by DISCo, i.e., via the solution of problem \eqref{eq:Opt_QHO}. (b) The analytical solution $\psi(x,\tau)=e^{-\frac{1}{2}(\tau + x^2)}$. (c)~Relative error $\epsilon_{rel} = \frac{|\psi_\text{DISCo} - \psi_\text{analytical}|}{|\psi_\text{analytical}|}$.}
    \label{fig:Control_QHO_optim}
\end{figure}

\paragraph{Example 2: The hydrogen atom.}

As a more challenging example, we consider the hydrogen atom, with initial condition given again by the analytical ground state. The state space is $\mathbb{X} = \mathbb{R}^3$, while
\begin{equation*}
    W(x)=-\frac{1}{r}=-\frac{1}{\norm{x}} \qquad \mbox{and} \qquad \psi_0(x) = e^{-\norm{x}} \quad \Rightarrow \quad \log(\psi_0(x)) = -\norm{x}.
\end{equation*}
The analytical solution is $\psi(x,\tau) = e^{-\frac{1}{2}\tau}e^{-\norm{x}}$. With this, the objective function in \eqref{eq:optimal_control_problem} and---using the Ornstein--Uhlenbeck process as described above---\eqref{eq:Objective_Koopman} becomes
\begin{equation}\label{eq:objective_hydrogen}
\mathbb{E}^x \left[ \int_\tau^{T} \left\{-\frac{1}{\norm{X_s^u}} +
\frac{1}{2}\ts\mathbb{E}^x\left[(X_s^u)^\top X_s^u\right] - \mathbb{E}^x \left[(X_s^u)\right]^\top\nu(s)+ \frac{1}{2}\nu^\top(s)\ts\nu(s) \right\} \, \mathrm{ds} +\norm{X_T^u} \right].
\end{equation}
Thus, the dictionary has to contain, in addition to monomial terms up to degree at least two, entries for $\frac{1}{\norm{X_s^u}}$ as well as $\norm{X_s^u}$, i.e.,
\begin{align*}
    \Phi(x) &= \begin{bmatrix}\frac{1}{\norm{x}} & \norm{x} & 1 & x_i & x_i\,x_j & \ldots \end{bmatrix}^\top, \quad 1 \leq i,j \leq 3.
\end{align*}
In order to make an informed dictionary selection, we perform a cross validation over a range of possible dictionaries $\Phi$. These are constructed in the following way. Denote by $\hat\Phi_p$ the dictionary consisting of all monomials up to degree $p$. Then set
\begin{equation*}
	\Phi(x) = \begin{bmatrix}
		\hat\Phi_p(x) \\ \hat\Phi_{p_{inv}}(x)\frac{1}{\norm{x}}\\
		\hat\Phi_{p_{norm}}(x)\norm{x}
	\end{bmatrix} = \begin{bmatrix}
	\begin{bmatrix}	1 & x_i & x_i\,x_j & \ldots\end{bmatrix}^\top \\
	\begin{bmatrix} \frac{1}{\norm{x}} & \frac{x_i}{\norm{x}} & \frac{x_i\,x_j}{\norm{x}} & \ldots\end{bmatrix}^\top \\
	\begin{bmatrix} \norm{x} & x_i\norm{x} & x_i\,x_j\norm{x} & \ldots\end{bmatrix}^\top
	\end{bmatrix},
\end{equation*}
with $p$, $p_{inv}$, and $p_{norm}$ being variable parameters that define the respective degrees of the monomial dictionaries. Moreover, we consider the special cases $p_{inv}=\emptyset$ or $p_{norm} = \emptyset$, in which $\frac{1}{\norm{x}}$ (or $\norm{x}$, respectively) are not explicitly calculated, but approximated using entries from $\hat\Phi_p$ as follows:
\begin{equation*}
	\mathbb{E}\left[\frac{1}{\norm{x}}\right] \approx \frac{1}{\norm{\mathbb{E}[x]}}\qquad \text{and} \qquad \mathbb{E}\left[\norm{x}\right] \approx \norm{\mathbb{E}[x]}.
\end{equation*}
Even though these approximations can be very coarse, the results are convincing in practice since the prediction accuracy appears to be generally higher for classical monomial dictionaries.

Denoting by $I_{1}$, $I_{2}$, $I_{inv}$, and $I_{norm}$ the index sets for the entries corresponding to the identity, squared, inverted norm and norm terms, i.e.,
\begin{align*}
\Phi_{I_1}(x) = \begin{bmatrix}x_{1} \\ x_{2} \\ x_{3}\end{bmatrix}, \quad
\Phi_{I_2}(x) = \begin{bmatrix}x^2_{1} \\ x^2_{2} \\ x^3_{3}\end{bmatrix}, \quad
\Phi_{I_{inv}}(x) = \frac{1}{\|x\|}, \quad
\Phi_{I_{norm}}(x) = \|x\|,
\end{align*}
the objective function \eqref{eq:objective_hydrogen} becomes
\begin{equation} \label{eq:objective_hydrogen2}
    \int_\tau^{T} \left\{ - z_{I_{inv}}(s) +\frac{1}{2} \sum_{j\in I_2} (z_{I_{2}})_j (s) - (z_{I_{1}}(s))^\top\nu(s) + \frac{1}{2}(\nu(s))^\top \nu(s) \right\} \, \mathrm{ds} + z_{I_{norm}}(T),
\end{equation}
and the resulting control problem we need to solve then becomes
\begin{equation} \label{eq:Opt_Hydrogen}
    \min_{\nu \colon [\tau, T] \to \mathbb{R}^d} \eqref{eq:objective_hydrogen2}
\qquad s.t.\quad \eqref{eq:ODE_Koopman}~\text{with initial condition}~z(\tau)=\Phi(x).
\end{equation}

\begin{figure}
    \centering
    \subfiguretitle{(a) \hspace*{6.5cm} (b)}
    \includegraphics[width=0.95\textwidth]{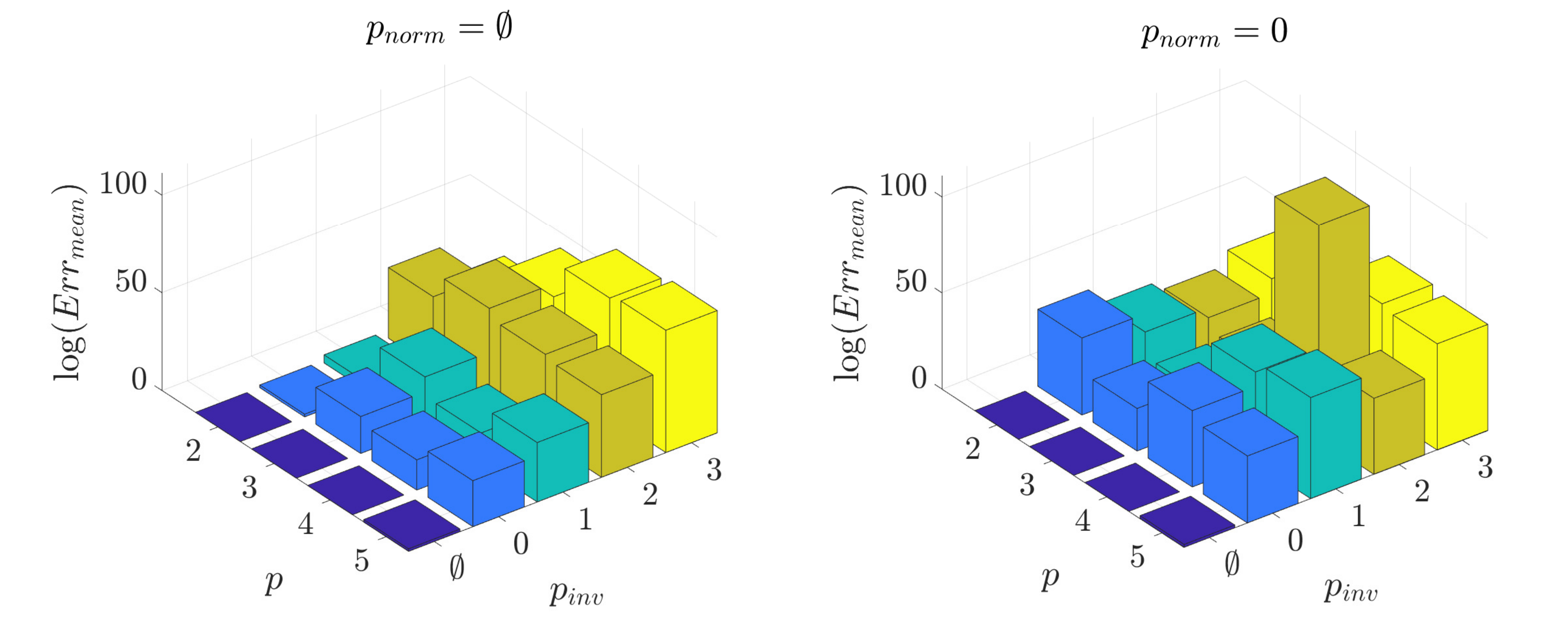}
    \caption{Average prediction performance (averaged over 1000 controlled trajectories and 10 model realizations) for different dictionary sizes $p$, $p_{inv}$, and $p_{norm}$. The best configuration turns out to be $p=2$ and $p_{inv}=p_{norm}=\emptyset$.}
    \label{fig:Hydrogen_ModelSelection}
\end{figure}

Before we approach this, we first perform a study over various degrees for $p$, $p_{inv}$, and $p_{norm}$. To this end, we train ten models for each parameter set and then validate the prediction accuracy of the resulting models against 1000 trajectories with random sinusoidal inputs with uniformly distributed parameters:
\begin{equation*}
    u(s) = \begin{bmatrix}
    a_1 \sin(b_1 \ts s + c_1) \\
    a_2 \sin(b_2 \ts s + c_2) \\
    a_3 \sin(b_3 \ts s + c_3)
    \end{bmatrix}, \qquad \mbox{where}~a_i \propto U(0,5),\quad b_j \propto U(2\pi,6\pi),\quad c_k \propto U(0,2\pi).
\end{equation*}
The error is then calculated as the average $L_2$ error between the prediction of the Koopman model for the terms relevant in the objective function and the expected value of the true system, approximated by averaging over 1000 simulations using the Euler--Maruyama scheme. Figure~\ref{fig:Hydrogen_ModelSelection} shows this analysis, where on average the best configuration turns out to be $p=2$ and $p_{inv}=p_{norm}=\emptyset$, i.e., despite being mathematically inexact, it is beneficial to rely exclusively on polynomial basis functions.
Nevertheless, the best models including either the term $\frac{1}{\norm{x}}$ or $\norm{x}$ are quite similar, which is shown for an exemplary test trajectory in Figure \ref{fig:Hydrogen_ModelComparison}. Still, the model with $p_{inv}=p_{norm}=\emptyset$ (third column) appears to be slightly superior.

\begin{figure}
	\centering
	\includegraphics[width=\textwidth]{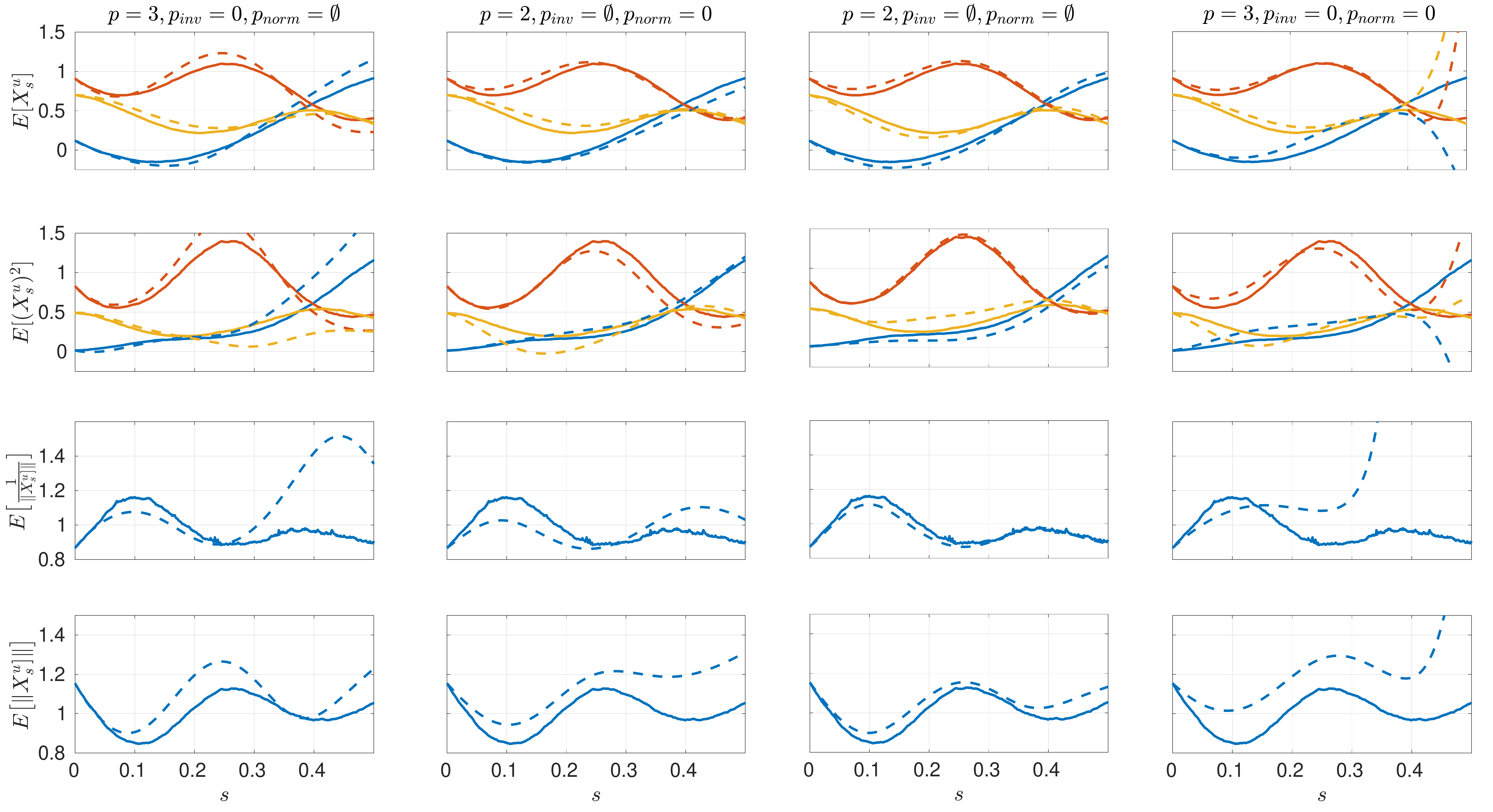}
	\caption{Comparison between expectations calculated by taking averages over 1000 Ornstein--Uhlenbeck simulations (solid lines) and the predictions obtained from the best data-driven models (dashed lines) out of ten repetitions for different sizes of $p$, $p_{inv}$, and $p_{norm}$. If the terms $\frac{1}{\norm{x}}$ or $\norm{x}$ are not explicitly contained in the dictionary, then we approximate them as described above.}
	\label{fig:Hydrogen_ModelComparison}
\end{figure}

Based on these conclusions, we proceed with the model with $p=2$ and $p_{inv}=p_{norm}=\emptyset$ for simulating the imaginary-time Schrödinger equation by solving \eqref{eq:Opt_Hydrogen}. Figure \ref{fig:Control_Hydrogen} shows a comparison between the normalized DISCo solution \eqref{eq:Opt_Hydrogen} with the analytical solution for $\tau = 0.5$ and 1000 randomly drawn points from the sphere with radius 2. Despite minor inaccuracies, we observe a very good agreement, with an average error of approximately~$4.5\%$.

\begin{figure}
    \centering
    \subfiguretitle{\hspace*{1.cm} (a) \hspace*{5.5cm} (b)}
    \includegraphics[width=.8\textwidth]{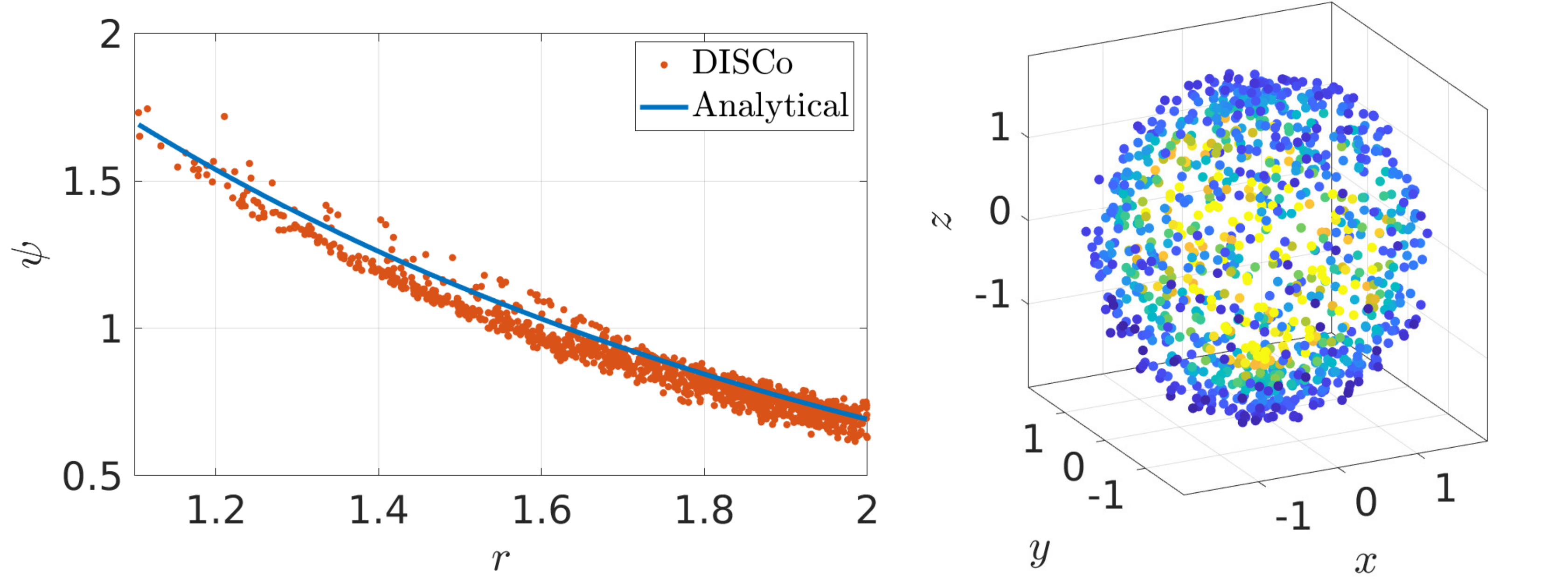}
    \caption{(a) Comparison between analytic and DISCo solutions of the imaginary-time Schrödinger equation at $\tau=0.5$ for 1000 i.i.d.\ initial conditions from the interior of a sphere with radius 2. (b) Visualization of the DISCo solution from (a), where coloring denotes the value of the wave function $\psi$.}
    \label{fig:Control_Hydrogen}
\end{figure}

\section{Conclusion}
\label{sec:Conclusion}

We have shown how Koopman operator theory can be applied to quantum mechanics problems, either by transforming the Schrödinger equation into a Kolmogorov backward equation or by using Nelson's stochastic mechanics, and how data-driven methods for the approximation of the Koopman operator can be used to analyze quantum systems. The analysis of non-stationary problems requires the approximation of forward-backward operators, which was illustrated with the aid of a superposition of wave functions. Moreover, we pointed out relationships between the imaginary-time Schrödinger equation and stochastic optimal control problems. This allowed us to exploit Koopman-based control techniques in order to solve the Schrödinger equation by means of an optimal control problem constrained by an ordinary differential equation. We presented numerical results for various benchmark problems such as the quantum harmonic oscillator, the Pöschl--Teller potential, and the hydrogen atom.

We have seen in Section~\ref{sec:Data-driven methods for quantum systems} that it is in principle possible to directly apply DMD and related techniques to quantum systems. However, this approach suffers from the curse of dimensionality as direct integration of the time-dependent Schrödinger equation is still required. With this in mind, methods that rely only on time-series data generated by associated stochastic differential equations seem particularly promising. Future work will have to focus on determining tailor-made dictionaries for quantum systems, and on the use of kernel methods or deep learning architectures in this context. Bohmian mechanics may also open up an interesting avenue for the analysis of quantum systems using particle trajectories.

The proposed DISCo approach to solve the imaginary-time Schrödinger equation via deterministic optimal control problems presents a novel and entirely data-driven approach to this long-standing problem. We have shown that a simple Ornstein--Uhlenbeck-type model seems to work well as a choice of control system. Assessing its suitability to model more complex systems, and the investigation of different types of control systems, will be one of the most pressing open questions in this context, next to the definition of bespoke dictionaries mentioned above.

\section*{Acknowledgments}
\textbf{Funding:} SP acknowledges support by the Priority Programme 1962 of the Deutsche Forschungsgemeinschaft (DFG).

\bibliographystyle{unsrturl}
\bibliography{Nelson}

\appendix

\section{Continuity equation}
\label{app:Continuity equation}

We start with the probability density $ p(x, t) $ associated with \eqref{eq:SDE}. This density solves the Fokker--Planck equation, which can be written as a continuity equation, i.e.,
\begin{equation*}
    \pd{p}{t} = -\nabla \cdot \left[ \left(b - \frac{1}{2p} \sum_{j=1}^d \pd{}{x_j} (a_{ij} \ts p)\right)\ts p \right].
\end{equation*}
If, moreover, the diffusion $\sigma$ in~\eqref{eq:SDE} is a constant multiple of the identity, the continuity equation simplifies to
\begin{equation*}
    \pd{p}{t} = - \nabla \cdot \left[ \left(b - \frac{\sigma^2}{2p}\nabla p \right)\ts p \right] = - \nabla \cdot \left[ \left(b - \frac{\sigma^2}{2} \nabla \log p \right)\ts p \right].
\end{equation*}
On the other hand, assume that $ \psi =  e^{R + i \ts S} $ solves \eqref{eq:TDSE}. Then
\begin{equation*}
    i \left(\pd{R}{t} + i \pd{S}{t} \right) = -\frac{1}{2} \big(\nabla R \cdot \nabla R - \nabla S \cdot \nabla S + 2 \ts i \ts \nabla R \cdot \nabla S + \Delta R + i \Delta S\big) + W.
\end{equation*}
For the imaginary part, this yields
\begin{equation*}
    \pd{R}{t} = -\nabla R \cdot \nabla S - \frac{1}{2} \Delta S.
\end{equation*}
The quantum probability distribution, given by $ \rho = e^{2 \ts R} $, thus satisfies
\begin{equation*}
    \pd{\rho}{t} = 2 \pd{R}{t} \ts \rho = \big(-2 \ts \nabla R \cdot \nabla S - \Delta S\big) \rho = -\nabla \cdot \left[\nabla S \ts \rho\right] = -\nabla \cdot \left[v \ts \rho\right].
\end{equation*}
Consequently, we will have $ p = \rho $ if the drift term $ b $ and diffusion constant $ \sigma $ are chosen such that
\begin{equation*}
    b - \frac{\sigma^2}{2} \nabla \log \rho = b - \sigma^2 \ts \nabla R = \nabla S,
\end{equation*}
which will be satisfied, e.g., for
\begin{equation*}
    \sigma^2 = 1
    \quad \text{and} \quad
    b = \nabla R + \nabla S = u + v.
\end{equation*}
Note that by setting $ \sigma^2 = 0 $, we obtain Bohmian mechanics.

\section{Bellman equation for stabilized control problem}
\label{app:stabilized_control}

We verify the claim made in Section~\ref{subsec:solution_control_problem} that if the control policy in~\eqref{eq:optimal_control_problem} is written as $u(s) = G(X_s^u) \ts \nu(s)$, then a strong solution $V$ to the HJB equation~\eqref{eq:hjb_pde} equals the value function of the stabilized control problem
\begin{equation} \label{eq:control_problem_modified}
\begin{split}
J(x, \tau) = \min_{\substack{\nu \colon [\tau, T] \to \mathbb{R}^d \\ \nu \in \mathcal{U}}} \, & \mathbb{E}^x \left[ \int_\tau^{T} W(X_s^u) + \frac{1}{2} \|G(X_s^u)\nu(s) \|^2 \, \mathrm{ds} - \log(\psi_0(X_T^u)) \right] \\
\text{s.t.} \quad \mathrm{d}X_s^u &= G(X_s^u)\nu(s) \, \mathrm{d}s + \mathrm{d}B_s, \quad X_\tau^u = x.
\end{split}
\end{equation}
To this end, we consider the dynamic programming equation~\cite{Fleming2006}[Ch. IV.3]
\begin{align*}
-\pd{V(x, \tau)}{\tau} &= \inf_{\nu \in \mathcal{U}} \left[\frac{1}{2}\Delta V(x, \tau) + W(x) + \nu^T G(x)^T \cdot \nabla V(x, \tau) +\frac{1}{2}\|G(x)\nu \|^2\right], \\
\nonumber V(x, T) &= -\log(\psi_0(x)).
\end{align*}
As the first two terms terms are independent of $\nu$, it suffices to minimize the remaining two terms with respect to $\nu$, which leads to the criticality condition:
\begin{equation*}
    (G(x)^\top G(x)) \ts \nu = - G(x)^\top \nabla V(x, \tau).
\end{equation*}
This equation is satisfied for any $\nu^* \in \mathbb{R}^{d_u}$ such that $G(x) \ts \nu^* = -\nabla V(x, \tau)$. As $G(x)$ is full rank, at least one such $\nu^*$ exists, and the minimal value attained is
\begin{equation*}
    (\nu^*)^\top G(x)^\top \cdot \nabla V(x, \tau) + \frac{1}{2}\|G(x)\nu^* \|^2 = -\frac{1}{2} \|\nabla V(x,\tau)||^2.
\end{equation*}
Thus, the dynamic programming equation is equivalent to the HJB equation~\eqref{eq:hjb_pde}.

\section{Superposition of two wave functions}
\label{app:Superposition}

For a superposition of two wave functions $ \psi = \psi_1 + \psi_2 $ with $ \psi_j = e^{R_j + i \ts S_j} $, we obtain the current and osmotic velocities $ u =  \frac{u'}{w} $ and $ v = \frac{v'}{w} $, with
\begin{align*}
    u' &= e^{2 R_1} \nabla R_1 + e^{2 R_2} \nabla R_2 \\
    & + e^{R_1 + R_2} \left[ \cos(S_1 - S_2) (\nabla R_1 + \nabla R_2) - \sin(S_1 - S_2)(\nabla S_1 - \nabla S_2) \right], \\
    v' &= e^{2 R_1} \nabla S_1 + e^{2 R_2} \nabla S_2 \\
    & + e^{R_1 + R_2} \left[\sin(S_1 - S_2) (\nabla R_1 - \nabla R_2) + \cos(S_1 - S_2)(\nabla S_1 + \nabla S_2)\right], \\
    w &= e^{2 R_1} + e^{2 R_2} + 2 \ts e^{R_1 + R_2} \cos(S_1 - S_2).
\end{align*}
We use such a superposition of two wave functions to construct a system with time-dependent dynamics in Section~\ref{sec:Data-driven methods for quantum systems}.

\end{document}